\documentclass{fundam}
\usepackage{url} 
\usepackage{graphicx}

\usepackage{setspace}
\usepackage{graphicx}
\usepackage{amsfonts}
\usepackage{amsmath}
\usepackage{ amssymb }
\usepackage{tikz}
\usetikzlibrary{shapes.geometric}
\usetikzlibrary{arrows,shapes,snakes,automata,backgrounds,petri,positioning}
\usepackage{rotating}
\usepackage{stmaryrd}
\usetikzlibrary{positioning,calc,decorations,decorations.pathreplacing,petri,patterns}
\usepackage{enumitem}
\usepackage{thmtools}
\usepackage{thm-restate}

\newcommand{\nat}{\mathbb{N}}

\newcommand{\N}{\mathcal N}

\newcommand{\EXPSPACE}{{\sf EXPSPACE}}
\newcommand{\NEXPSPACE}{{\sf NEXPSPACE}}

\newcommand{\EXPTIME}{{\sf EXPTIME}}

\newcommand{\TOWER}{{\sf TOWER}}

\newcommand{\at}{abstract}
\newcommand{\lseq}{deep}
\newcommand{\sseq}{shallow}

\newcommand{\bigslant}[2]{{\raisebox{0em}{$#1\!$}\left/\raisebox{-.1em}{$\!#2$}\right.}}

\newcommand{\rooted}[1]{\mathring{#1}}

\newcommand{\sr}{{\mathring{s_0}}}
\newcommand{\abst}[1]{{[{#1}]}}

 \usepackage{hyperref}

\begin{document}

\setcounter{page}{33}
\publyear{2021}
\papernumber{2081}
\volume{183}
\issue{1-2}

  \finalVersionForARXIV

\title{Coverability, Termination, and Finiteness in Recursive Petri Nets}

\address{I. Khmelnitsky, LSV, ENS Paris-Saclay}

\author{Alain Finkel\thanks{The work of this author was carried out
		in the framework of ReLaX, UMI2000 and also supported by ANR-17-CE40-0028
		project  BRAVAS.}
 \\
LSV, ENS Paris-Saclay, CNRS, IUF, \href{https://orcid.org/0000-0003-2482-6141}{ORCID}\\
Universit\'e Paris-Saclay, Gif-sur-Yvette, France\\
alain.finkel@ens-paris-saclay.fr
\and
     Serge Haddad\thanks{The work of this author
	was partly supported by ERC project EQualIS (FP7-308087)}, $\,$Igor Khmelnitsky
\\
LSV, ENS Paris-Saclay, CNRS, INRIA\\
Universit\'e Paris-Saclay, Gif-sur-Yvette, France\\
\{serge.haddad,  igor.khmelnitsky\}@ens-paris-saclay.fr
 }

\maketitle

\runninghead{A. Finkel et al.}{Coverability, Termination,  and Finiteness in RPN}

\begin{abstract}
In the early two-thousands, Recursive Petri nets have been introduced
in order to model distributed planning of multi-agent systems for which counters
and recursivity were necessary.
Although Recursive Petri nets strictly extend Petri nets and context-free grammars,
most of the usual problems (reachability, coverability, finiteness, boundedness and termination) were known to be solvable
by using non-primitive recursive  algorithms.
For almost all other extended Petri nets models
containing a stack,
the complexity of coverability and termination are unknown
or strictly larger than \EXPSPACE. In contrast, we establish here that for Recursive Petri nets,
the coverability, termination, boundedness and finiteness problems are \EXPSPACE-complete as for Petri nets.
From an expressiveness point of view,
we show that coverability languages of Recursive Petri
nets strictly include the union of coverability languages
of Petri nets and context-free languages.
Thus we get a more powerful model than Petri net for free.
\end{abstract}

\section{Introduction}
\label{sec:introduction}

{\bf Verification problems for Petri nets.}
Petri net is a useful formalism for the analysis of concurrent programs
for several reasons. From a modeling
point of view (1) due to the locality of the firing rule, one easily models
concurrent activities and (2) the (a priori) unbounded marking of places allows
to represent a dynamic number of activities. From a verification point of view,
most of the usual properties are decidable. However, Petri nets suffer
two main limitations: they cannot model recursive features and the computational
cost of verification may be very high. More precisely,
all the known algorithms solving reachability
are nonprimitive recursive (see for instance~\cite{Mayr84})
and it has been proved recently that the reachability
problem is non elementary~\cite{abs-1809-07115} but primitive recursive when
the dimension is fixed~\cite{DBLP:conf/lics/LerouxS19}.
Fortunately some interesting properties like coverability,
termination, finiteness, and boundedness are \EXPSPACE-complete~\cite{Rac78}
and thus still manageable by a tool. So an important research direction
consists of extending Petri nets to support new modeling features
while still preserving decidability of properties checking and if possible with a "reasonable" complexity.

\noindent
{\bf Extended Petri nets.}
Such extensions may be partitioned between those whose states are still markings
and the other ones.
The simplest extension consists of adding inhibitor arcs
which yields undecidability of most of the verification problems.
However adding a single inhibitor arc preserves the decidability of
the reachability, coverability, and boundedness problems~\cite{Reinhardt08,BFLZ-lmcs12,Bonnet11}.
When adding reset arcs, the coverability problem becomes
Ackermann-complete ~\cite{PhS-mfcs10} and boundedness undecidable~\cite{dufourd98}

In $\nu$-Petri nets, the tokens are colored where colors are picked in
an infinite domain: their coverability problem is double-Ackermann
time
complete~\cite{lazic:hal-01265302}.
In Petri nets with a stack, the reachability problem may be reduced
to the coverability problem and both are at least not elementary ~\cite{abs-1809-07115,Lazic13}
while their decidability status is still unknown~\cite{Lazic13}.
In branching vector addition systems with states (BVASS)
a state is a set of threads with associated markings.
A thread either fires a transition as in Petri nets or forks, transferring
a part of its marking to the new thread. For BVASS,
the reachability problem is also
\TOWER-hard~\cite{LazicS14}
and its decidability is still an open problem while the coverability
and the boundedness problems are  2-\EXPTIME-complete~\cite{jcss12-DJLL}.
The analysis of subclasses of Petri nets
with a stack is an active field of research~\cite{AtigG11,MavlankulovOTSZ18,DassowT09,Zetzsche15}.
However, for none of the above extensions, the coverability and termination problems
belong to \EXPSPACE.

\noindent
{\bf Recursive Petri nets (RPN).} This formalism has been introduced
to model distributed planning of multi-agent systems for which counters and recursivity were necessary for specifying resources
and delegation of subtasks~\cite{EFH-icmas96}.
Roughly speaking, a state of an RPN consists of a tree of \emph{threads}
where the local state of each thread is a marking.
Any thread fires an \emph{elementary}, \emph{abstract} or \emph{cut} transition. When the transition is elementary,
the firing updates its marking as in Petri nets; when it is abstract, this only consumes
the tokens specified by the input arcs of the transition
and creates a child thread initialized with the \emph{initial marking} of
the transition. When a cut transition is fired,
the thread and its subtree are pruned, producing in its parent
the tokens specified by the output arcs of the abstract transition that created it.
In RPN, reachability, boundedness and termination
are decidable~\cite{HP-icatpn99,haddad:hal-01573071}
by reducing these properties to reachability problems of Petri nets. So the corresponding algorithms are nonelementary.
LTL model checking is undecidable for RPN but becomes decidable for the subclass of sequential RPN~\cite{HaddadP01}.
In~\cite{HaddadP07}, several modeling features are proposed while preserving the decidability of
the verification problems.


\noindent
{\bf Our contribution.}
We first study the expressive power of RPN from the point of view of coverability
languages (reachability languages were studied in \cite{HP-icatpn99}).
We first introduce a quasi-order on states of RPN compatible with the firing rule and establish that it is not a well quasi-order. Moreover, we show that there cannot exist a transition-preserving compatible well quasi-order, preventing us to use the framework of Well Structured Transition Systems
to prove that coverability is decidable.
We show that the RPN languages
are \emph{quite close} to recursively enumerable languages since the closure under homomorphism and intersection with a regular language is the family of recursively enumerable languages.
More precisely, we show that  RPN coverability
(as reachability) languages strictly include the union of context-free languages
and Petri net coverability languages.
Moreover, we prove that  RPN coverability
languages and reachability languages of Petri nets
are incomparable. We prove that RPN coverability languages are a strict subclass of RPN reachability languages.
In addition, we establish that
the family of RPN languages is closed under union, homomorphism
but neither under intersection with a regular language nor under complementation.

From an algorithmic point of view, we show that, as for Petri nets, coverability, termination,
boundedness, and finiteness are
\EXPSPACE-complete. Thus the increase of expressive power does
not entail a corresponding increase in complexity.
In order to solve the coverability problem, we show that if there exists
a covering sequence there exists a `short' one
(i.e. with a length at most doubly exponential w.r.t. the size of the input).
In order to solve the termination problem, we
	consider two cases for an infinite sequence
	depending (informally speaking) whether the depth of the trees
	corresponding to states are bounded or not along the sequence. For the unbounded case, we introduce the abstract graph that expresses the ability to create threads from some initial state.
The decidability of the finiteness and boundedness problems
are also mainly based on this abstract graph.

Let us mention that this paper is an extended version of \cite{FHK-atpn19} that contains new results about expressiveness like the characterization of the RPN coverability languages, decidability and complexity of finiteness and boundedness and we greatly simplified the proofs of coverability, termination, and finiteness. We also provided a more elegant definition of the (now inductive) syntax and the semantics of RPN.

\noindent
{\bf Outline.} In section~\ref{sec:recursive}, we introduce RPNs and
state ordering and establish basic results related to these notions.
In section~\ref{sec:reductions}, we introduce decision problems and some reductions between them.
In section~\ref{sec:expressiveness}, we study the expressiveness
of coverability languages. Then in
sections~\ref{sec:coverability},~\ref{sec:termination}, and~\ref{sec:finiteness}
we show that the coverability, termination, boundedness, and finiteness problems are \EXPSPACE-complete.
In section~\ref{sec:conclusion}, we conclude and give some perspectives to this work.

\global\long\def\P{\mathbb{P}}
\global\long\def\R{\mathbb{R}}

\section{Recursive Petri nets}
\label{sec:recursive}

\subsection{Presentation}

The state of an RPN has a structure akin to a `directed rooted tree' of Petri nets. Each vertex of the tree,
hereafter \emph{thread}, is an instance of the RPN and possessing some marking on it.
Each of these threads can fire \emph{three} types of transitions.
An \emph{elementary} transition updates its own marking according
to the usual Petri net firing rule.
An \emph{abstract} transition consumes tokens from the thread firing it
and creates a new child (thread) for it. The marking of the new thread is determined
according to the fired abstract transition.
A \emph{cut} transition can be fired by a thread
if its marking is greater or equal
than some marking.
Firing a cut transition, the thread erases itself and all of its descendants.
Moreover, it creates tokens in its parent,
which are specified by the abstract transition that created it.

\begin{definition}[Recursive Petri Net]
	\noindent A \emph{Recursive Petri Net} is a 6-tuple
	$\N=\langle P,T,W^{+},W^{-},\Omega\rangle$ where:
	\begin{itemize}
		\item $P$ is a finite set of places;
		\item $T=T_{el}\uplus T_{ab} \uplus T_{\tau}$ is a finite set of transitions with $P\cap T=\emptyset$,
		and $T_{el}$ (respectively $T_{ab}, T_{\tau}$) is the subset of elementary (respectively abstract, cut) transitions;
		\item $W^{-}$ is the $\nat^{P\times T}$ backward incidence matrix;
		\item $W^{+}$ are the $\nat^{P\times (T_{el}\uplus T_{ab})}$ forward incidence matrix;
		\item $\Omega : T_{ab} \rightarrow \nat^P$ is a function that labels every abstract transition with a initial marking;
	\end{itemize}
\end{definition}

\begin{figure}[!b]
	\begin{center}
		\begin{tikzpicture}[]

\tikzstyle{place}=[circle,thick,draw=black!75,minimum size=5mm]
\tikzstyle{thread}=[circle,thick,draw=black!75,minimum size=3mm]
\tikzstyle{elementary_transition}=[rectangle,minimum width=15, minimum height=1,draw=black!75,]
\tikzstyle{abstract_transition}=[rectangle,minimum width=15, minimum height=0.5,double,draw=black!100,]
\tikzstyle{abstract_transition_description}=[rectangle,minimum width=5, minimum height=5,draw=black!75,]
\usetikzlibrary{arrows,shapes,snakes,automata,backgrounds,petri}
\tikzstyle{cut_transition}=[rectangle,minimum width=15, minimum height=0.5,fill=black,draw=black!100,]

\node [place, tokens=0,label=150:$p_{beg}$] (pbeg)  {};

\node [place, tokens=0,label=180:$p_{a_1}$] (pa1)  [below left = 11mm and 8mm of pbeg]  {};
\node [place,tokens=0,label=180:$p_{a_2}$] (pa2)  [below = 10mm  of pa1] {};

\node [place,tokens=1,label=200:$p_{end}$] (pe) [below  = 39mm of pbeg]  {};

\node [place, tokens=0,label=0:$p_{b_1}$] (pb1)  [below right = 11mm and 8mm of pbeg]  {};
\node [place,label=0:$p_{b_2}$] (pb2)  [below = 10mm  of pb1] {};

\node [place, tokens=0,label=$p_{ini}$] (pini) [right = 35mm of pbeg] {};
\node [place, tokens=0,label=270:$p_{fin}$] (pfin) [below = 15mm of pini] {};

%
\node  [cut_transition,label=0:$t_{\tau_1}$] (ta2) [above = 5mm of pbeg] {}
edge[pre] (pbeg);
\node  [cut_transition,label=0:$t_{\tau_2}$] (ta2) [below = 5mm of pe] {}
edge[pre] (pe);

\node [elementary_transition,label=180:$t_{a_1}$] (ta1) [below left = 3mm and 7.5mm  of pbeg] {}
	edge[pre,bend left] (pbeg)
	edge[post] (pa1);

\node  [abstract_transition,label=180:$t_{a_2}$] (ta2) [below = 3.5mm of pa1] {}
	edge[pre] (pa1)
	edge[post] (pa2);
\node [abstract_transition_description] [left = 6mm of ta2]  {$p_{beg}$};

\node [elementary_transition,label=180:$t_{a_3}$] (ta3) [below = 3.5mm of pa2] {}
	edge[pre] (pa2)
	edge[post,bend right] (pe);
\node [elementary_transition,label=0:$t_{b_1}$] (tb1) [below right = 3mm and 7.5mm  of pbeg] {}
edge[pre,bend right] (pbeg)
edge[post] (pb1);

\node  [abstract_transition,label=0:$t_{b_2}$] (tb2) [below = 3.5mm of pb1] {}
edge[pre] (pb1)
edge[post] (pb2);
\node [abstract_transition_description] [right = 6mm of tb2]  {$p_{beg}$};

\node [elementary_transition,label=0:$t_{b_3}$] (tb3) [below = 3.5mm of pb2] {}
edge[pre] (pb2)
edge[post,bend left] (pe);

\node [elementary_transition,label=180:$t_{sb}$] (tsb) [below right = 26mm and 1mm  of pbeg] {}
edge[pre] (pbeg)
edge[post] (pe);

\node [elementary_transition,label=0:$t_{sa}$] (tsb) [below left = 13mm and 1mm  of pbeg] {}
edge[pre] (pbeg)
edge[post] (pe);



\node  [abstract_transition,label=180:$t_{beg}$] (ta2) [below  = 6mm of pini] {}
edge[pre] (pini)
edge[post] (pfin);
\node [abstract_transition_description] [right = 2mm of ta2]  {$p_{beg}$};

{\small
\node[thread, label=0:$\bf0$ (the root $ r_s $) ] (v1) [below right   = 0mm and 30mm of pini]  {};
\node[thread, label=0:   $\bf0$] (v2) [below  = 8mm of v1] {};
\node[thread, label=0:   $p_{end}$] (v3) [below  = 8mm of v2] {};

\path[->] (v1) edge node[label=0:$p_{fin}$]  {} (v2);
\path[->] (v2) edge node[label=0:$p_{b_2}$]  {} (v3);
\path[dashed] (v3) edge node[label=0:$ $]  {} (56mm,-57mm);
\path[dashed] (v3) edge node[label=0:$ $]  {} (56mm,12mm);
\node [above right  = 2mm and -10mm of v1] {$ s $ - A state of $ \N $};
\node [right  = 5mm of v1] {};
}


\draw[color=black,dashed](55mm,-57mm) rectangle (-30mm,12mm);
\node at (-23mm,7mm) {\color{black}RPN $\N$};

\end{tikzpicture}
	\end{center}\vspace*{-3mm}
	\caption{An example of a marked RPN.}
	\label{fig:rpn_example}
\end{figure}
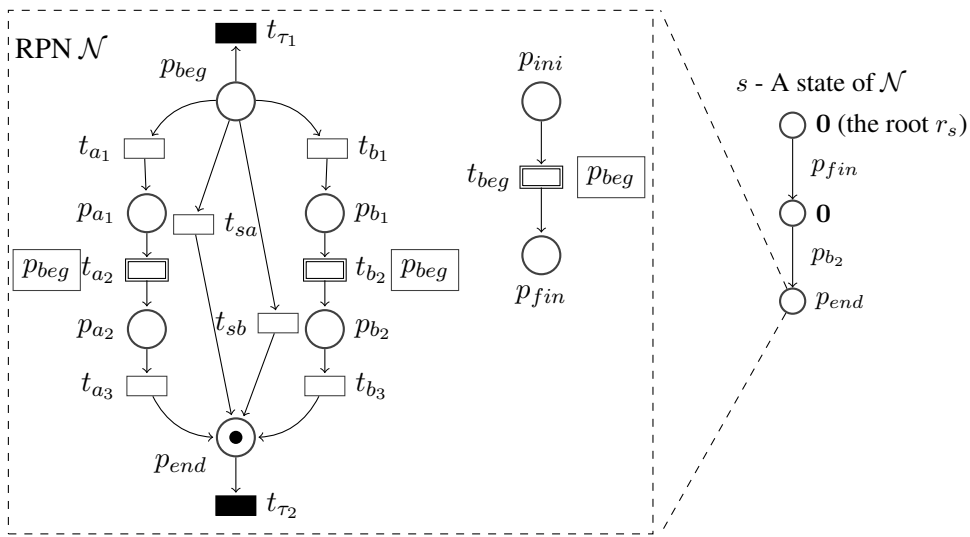

\noindent Figure~\ref{fig:rpn_example} graphically describes an example of an RPN  with:
\begin{align*}
	&P=\{p_{ini},p_{fin},p_{beg},p_{end}\}\cup\{p_{b_i},p_{a_i}:i\leq 2 \}; \\
	&T_{el}=\{t_{b_1},t_{b_3},t_{a_1},t_{a_3},t_{sa},t_{sb}\}\,;\,T_{ab}=\{t_{beg},t_{b_2},t_{a_2}\};\\
	&T_{\tau}=\{t_{\tau_1},t_{\tau_2}\}.
\end{align*}
and for instance $ W^-(p_{ini},t_{beg})=1 $ and $ \Omega(t_{b_2})= p_{beg} $ (where  $ p_{beg} $ denotes the marking with one token in place $ p_{beg} $ and zero elsewhere).

For brevity reasons, we denote by $W^+(t)$ a vector in $\nat^P$,
where for all $p\in P$, $W^+(t)(p)=W^+(p,t)$, and we do the same for $W^-(t)$.

A \emph{concrete state} $s$ of an RPN is a labeled
tree representing relations between threads and their associated markings.
Every vertex of $s$ is a thread and edges are labeled by abstract transitions.
We introduce a countable set $\mathcal V$ of vertices in order to pick new vertices when necessary.

\begin{definition}[State of an RPN]
	A \emph{concrete state (in short, a state)} $s$ of an RPN is a tree over the finite set of vertices $V_s\subseteq \mathcal V$, inductively defined as follows:
	\begin{itemize}
		\item either $V_s=\emptyset$ and thus $s=\emptyset$ is the empty tree;
		\item or $V_s=\{r_s\}\uplus V_1\uplus \ldots \uplus V_k$ with $0\leq k$ and
		$s=(r_s,m_0,\{(m_i,s_i)\}_{1\leq i\leq k})$ is defined as follows:
		\begin{itemize}
			\item $r_s$ is the root of $s$ labelled by a marking $m_{0} \in \nat^P$;
			\item For all $i\leq k$, $s_i$ is a state over $V_i\neq \emptyset$
			
			and there is an edge $r_s\xrightarrow{m_i}_s r_{s_i}$ with $m_i\in \{W^+(t)\}_{t\in T_{ab}}$.
		\end{itemize}
	\end{itemize}
\end{definition}
For all $u,v\in V_s$, one denotes $M_s(u)$ the marking labelling $u$
and when $u\xrightarrow{m}_s v$, one writes $\Lambda(u,v):=m$.
State $s_v$ is the (maximal) subtree of $s$ rooted in $v$.

While the set of vertices $V_s$ will be important for analyzing the behavior of a firing sequence in an RPN, one can omit it and get a more abstract representation of the state. Note that contrary to the previous definition where $ \{(m_i,s_i)\}_{1\leq i\leq k} $ was a set, in the following definition we need a multiset $Child_s$.

\begin{definition}[Abstract state of an RPN]
	An \emph{abstract state} $s$ of an RPN is inductively defined as follows:
	\begin{itemize}
		\item either $s=\emptyset$ is the empty set ;
		\item or $s=(m_s,Child_s)$ where $m_s \in \nat^P$
		and $Child_s$ is a finite multiset of pairs $(m',s')$
		
		where $m'\in \{W^+(t)\}_{t\in T_{ab}}$ and $s'$ is an abstract state different from $\emptyset$.
	\end{itemize}	
\end{definition}
Given a concrete state $s$, we denote by $\abst{s}$ its abstract state.
Except if explicitly stated, a state is a concrete state.

In the other direction, given an abstract state $s$, one recovers its set of concrete states by
picking an arbitrary set of vertices $V_s\subseteq \mathcal V$ of appropriate cardinality and, inductively, arbitrarily splitting $V_s$ between the root and the pairs $(m,s')$.

For example, on the right side of Figure~\ref{fig:rpn_example},
there is a  (concrete) state of the RPN $ \N$. This state consists of three threads with markings
$ \bf 0,0,$ and $p_{end}$ (where $\bf 0 $ is the \emph{null marking})
and two edges with the labels $ W^+(t_{beg})$ and $W^+(t_{b_2}) $.

Let $ s $ be a state of some RPN.
Every thread $u$ different from the root has an unique \emph{parent}, denoted by $prd(u)$.
The \emph{descendants} of a thread $u$ consists of threads
in the subtree rooted in $u$ including $u$ itself.
We denote this set by $Des_{s}(u)$.
For $m\in \nat^P$,
denote by $s[r,m]:=(r,m,\emptyset)$, the state consisting of a single vertex $r$ whose
marking is $m$.
As usual, two markings $m,m'\in \nat^P$, over a set of places $P$,
are partially ordered as follows:
$m\leq m'$ if for all places $p\in P$, $m(p) \leq m'(p)$.

\begin{definition}[Operational semantics]
	Let $s=(r,m_0,\{(m_i,s_i)\}_{1\leq i\leq k})$ be a state.
	Then the firing rule $s\xrightarrow{(v,t)}s'$ where $v\in V_s $ and $t\in T$ is inductively defined as follows:
	\begin{itemize}
		\itemsep=0.95pt
		\item Let $t\in T_{el}$ such that $W^-(t)\leq m_0$, then one has $s\xrightarrow{r,t} (r,m_0-W^-(t)+W^+(t), \{(m_i,s_i)\}_{i\leq k})$
		\item Let $t\in T_{ab}$ such that $W^-(t)\leq m_0$, then one has $s\xrightarrow{r,t} (r,m_0-W^-(t), \{(m_i,s_i)\}_{i\leq k+1}))$
		
		where $m_{k+1}=W^+(t)$, $s_{k+1}=s[v,\Omega(t)]$ with $v\in \mathcal V \setminus V_s$
		\item Let $t\in T_{\tau}$ such that $W^-(t)\leq m_0$, then one has $s\xrightarrow{r,t} \emptyset$
		\item Let $i\leq k$ such that $s_i \xrightarrow{v,t} s'_i$  \\
		if $s'_i=\emptyset$ then $s \xrightarrow{v,t} (m_0+m_i,\{(m_j,s_j)\}_{1\leq j\neq i\leq k})$
		
		else     $s \xrightarrow{v,t} (m_0,\{m_j,s_j\}_{1\leq j\neq i\leq k}\cup \{m_i,s'_i\})$
	\end{itemize}
	\end{definition}

Figure \ref{fig:firing_example} illustrates a sequence of  transition
firings in the RPN described by Figure~\ref{fig:rpn_example}. The first transition $t_{beg}\in T_{ab}$ is fired  by the root.
Its firing results in a state for which the root has a new child (denoted by $ v $)
and a new outgoing edge with  label $ p_{fin} $.
The marking of the root is decreased to ${\bf 0}$
and $v$ is initially marked by $ \Omega(t_{beg})=p_{beg} $.
The second firing is due to an elementary transition
$t_{b_1}\in T_{el}$ which is fired by $v$.
Its firing results in a state for which the marking of $v$ is changed to
$M_s'(v)=M_s(v)+W^+(t_{b_1})-W^-(t_{b_1})=p_{b_1}$.
The fifth transition
to be fired  is the cut transition $t_{\tau_2}$, fired by the thread
with the marking $p_{end}$ (denoted by $w$).
Its firing results in a state where the thread $ w $ is erased,
and the marking of its parent is increased by $ W^+(t_{b_2})=p_{b_2}$.

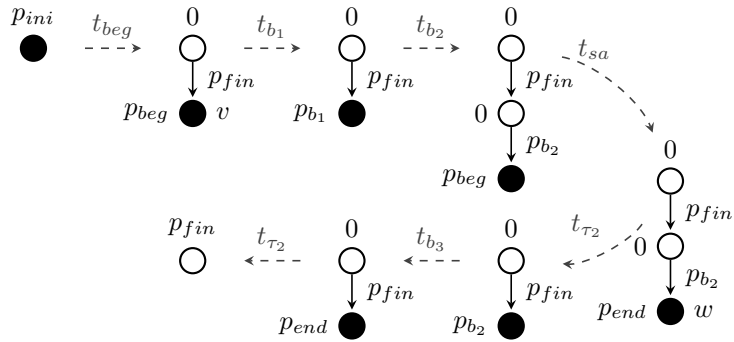
\begin{figure}[h]
\vspace{1mm}
		\begin{center}
		\begin{tikzpicture}[
			scale=0.7,
            > = stealth, 
            shorten > = 1pt, 
            auto,
            semithick 
        ]
\usetikzlibrary{arrows,shapes,snakes,automata,backgrounds,petri,positioning}
        \tikzstyle{state}=[
        	circle,
            draw = black,
            thick,
            fill = white,
            minimum size = 1mm
        ]
        \small
        \node[state, fill=black, label=   $p_{ini}$] (v11) at (0,0)  {};

        \node[state, label=   $0$] (v21) at (3,0)  {};
    	\node[state, fill=black, label=180:  $p_{beg}$,label=0: $v$] (v22) [below  = 5mm of v21] {};

        \path[->] (v21) edge node[xshift = -5,label=0:$  p_{fin}$]  {} (v22);

        \node[state, label=  $0 $] (v31) at (6,0)  {};
        \node[state,fill=black, label=180:   $p_{b_1}$] (v32) [below  = 5mm of v31] {};

        \path[->] (v31) edge node[xshift = -5,label=0:${  p_{fin}}$] {} (v32);

       \node[state, label=  $0 $] (v41) at (9,0)  {};
       \node[state, label=180:   $0$] (v42) [below  = 5mm of v41] {};
        \node[state,fill=black, label=180:   $p_{beg}$] (v43) [below  = 5mm of v42] {};

       \path[->] (v41) edge node[xshift = -5,label=0:${  p_{fin}}$] {} (v42);
       \path[->] (v42) edge node[xshift=-5,label=0:$p_{b_2}$] {} (v43);


        \node[state, label=  $0 $] (v51) at (12,-2.5)  {};
        \node[state, label=180:   $0$] (v52) [below  = 5mm of v51] {};
        \node[state,fill=black, label=180:   $p_{end}$,label=0:$w$] (v53) [below  = 5mm of v52] {};

        \path[->] (v51) edge node[xshift=-5,label=0:${p_{fin}}$] {} (v52);
        \path[->] (v52) edge node[xshift=-5,label=0:$p_{b_2}$] {} (v53);

        \node[state, label=  $0 $] (v61) at (9,-4)  {};
        \node[state,fill=black, label=180:   $p_{b_2}$] (v62) [below  = 5mm of v61] {};

        \path[->] (v61) edge node[xshift = -5,label=0: ${p_{fin}}$] {} (v62);

         \node[state, label=  $0 $] (v71) at (6,-4)  {};
         \node[state,fill=black, label=180:   $p_{end}$] (v72) [below  = 5mm of v71] {};

         \path[->] (v71) edge node[xshift=-5,label=0:${p_{fin}}$] {} (v72);

\node[state, label=  $p_{fin} $] (v81) at (3,-4)  {};


       \path[->,darkgray ,dashed,shorten <=0.5cm,shorten >=0.5cm] (v11) edge node {  $t_{beg}$} (v21);
       \path[->,darkgray,dashed,shorten <=0.5cm,shorten >=0.5cm] (v21) edge node {  $t_{b_1}$} (v31);
       \path[->,darkgray,dashed,shorten <=0.5cm,shorten >=0.5cm] (v31) edge  node {  $t_{b_2}$} (v41);
       \path[->,darkgray,dashed,shorten <=0.5cm,shorten >=0.5cm] (v41) edge [bend left] node[label=100:$t_{sa}$] {}(v51);
       \path[->,darkgray,dashed,shorten <=0.5cm,shorten >=0.5cm] (v51) edge [bend left] node[label=100:$t_{\tau_2}$] {}(v61);
       \path[->,darkgray,dashed,shorten <=0.5cm,shorten >=0.5cm] (v61) edge node[label=90:$t_{b_3}$] {}(v71);
       \path[->,darkgray,dashed,shorten <=0.5cm,shorten >=0.5cm] (v71) edge node[label=90:$t_{\tau_2}$] {}(v81);

\end{tikzpicture}
	\end{center}\vspace*{-4mm}
	\caption{Firing sequence for the RPN in Figure \ref{fig:rpn_example}}
	\label{fig:firing_example}\vspace*{-2mm}
\end{figure}

\smallskip
A \emph{firing sequence} is a sequence of transition firings, written
in a detailed way:
$s_0\xrightarrow{(v_1,t_1)}s_1\xrightarrow{(v_2,t_2)}\cdots\xrightarrow{(v_n,t_n)}s_n$,
or when the context allows it, in a more concise way like $s_0\xrightarrow{\sigma} s_n$ for
$\sigma=(v_1,t_1)(v_2,t_2)\dots(v_n,t_n).\,$Let $\sigma \in T^*$ with $\sigma\!=\!t_1\ldots t_n$ and $v$ be a vertex, $(v,\sigma)$
is an \mbox{abbreviation}\linebreak  for $(v,t_1)\ldots(v,t_n)$.
When we deal with several nets, we indicate by a subscript in which net, say $\N$,
the firing sequence takes place: $s_0\xrightarrow{\sigma}_{\N} s_n$.
Infinite firing sequences are similarly defined.
In a firing sequence, a thread $v$ that has been deleted is \emph{never reused} (which is possible
since $\mathcal V$ is countable).
A thread is \emph{final} (respectively \emph{initial}) w.r.t.
$\sigma$ if it occurs in the final (respectively initial) state of $\sigma$.
We say that $v\in Des_{\sigma}(u) $ if there exists $ i\leq n $ such that $ v\in Des_{s_i}(u) $.
We call $\sigma'$ a \emph{subsequence} of $\sigma$,
denoted by $\sigma'\sqsubseteq \sigma$, if there exists $k$ indexes $ i_1, i_2 \dots i_k$ such that $ 1 \leq i_1<i_2<\dots i_k\leq n $ and
$\sigma'=(v_{i_1},t_{i_1})(v_{i_2},t_{i_2})\dots (v_{i_k},t_{i_k})$.

\begin{remark}
	In the sequel, when we write ``RPN $\N$'', we mean $\N=\left<P,T,W^+,W^-,\Omega\right>$,
	unless we explicitly write differently. An RPN $\N$ equipped
	with an initial state $s$ is a \emph{marked~RPN} and denoted $(\N,s)$.
	Similarly a \emph{marked Petri net} $(\N,m)$ is a Petri net $\N$
	equipped with an initial marking $m$.
\end{remark}
For a marked RPN $(\N,s_0)$, let $Reach(\N,s_0)=\{\abst{s}\mid \exists\sigma\in T^* \text{ s.t. } s_0\xrightarrow{\sigma} s\}$ be its \emph{reachability set}, i.e. the set of all the reachable \emph{abstract} states.

\subsection{An order for Recursive Petri Nets}

We now define a quasi-order $\preceq $ on the states of an RPN.
Given two states $ s,s' $ of an RPN $\N$,
we  say that $ s $ is \emph{smaller or equal} than $ s' $, denoted by $s\preceq  s'$,
if there exists a subtree in $ s' $, which is isomorphic to $ s $,
where markings are greater or equal
on all vertices and edges.

\begin{definition}
	Let $s \neq \emptyset$ and $s'$ be states of an RPN $\N$.
	Then $ s\preceq  s'$ if there exists an injective mapping $f$ from $V_s$
	to $V_{s'}$ such that for all $v\in V_s$:\medskip

	\begin{enumerate}[nosep]
	\item $M_{s}(v) \leq M_{s'}(f(v))$, and,
		\item for all $v\xrightarrow{m}_sw$, there exists an edge $f(v) \xrightarrow{m'}_{s'}f(w)$ with $m\leq m'$.		
	\end{enumerate}\medskip\noindent
	In addition, $\emptyset \preceq s$ for all states $s$.
	
	\noindent When $f(r_s)$ is required to be $r_{s'}$, one denotes this relation $s\preceq_r  s'$
	with $\emptyset \preceq_r s$ if and only if $s=\emptyset$.	
\end{definition}

Figure~\ref{fig:order example} illustrates these quasi-orders.
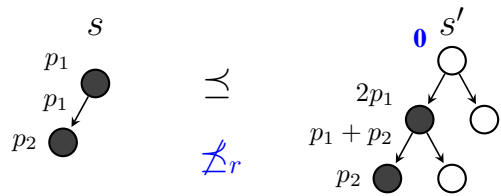
\begin{figure}[h]
	
	\begin{center}

\begin{tikzpicture}[
> = stealth, 
shorten > = 1pt, 
auto,
semithick 
]
\usetikzlibrary{arrows,shapes,snakes,automata,backgrounds,petri,positioning}
\tikzstyle{state}=[
circle,
draw = black,
thick,
fill = white,
minimum size = 1mm
]



\node[state, label=170: \small  {\bf\color{blue} 0}] (v51) at (0,0)  {};
\node[state, label=180: \small \color{blue}] (v52) [below right =5mm and 1.5mm of v51] {};
\node[state,fill= darkgray , label=170: \small  $2 p_1$] (v53) [below left = 5mm and 1.5mm of v51] {};
\node[state,fill= darkgray , label=180: \small $ p_2 $] (v54) [below left = 5mm and 1.5mm of v53] {};
\node[state, label=0: \small\color{blue} ] (v55) [below right = 5mm and 1.5mm of v53] {};

\path[->] (v51) edge node[] {} (v52);
\path[->] (v51) edge node {} (v53);
\path[->] (v53) edge node [label=160:\small$ p_1+p_2 $]  {}  (v54);
\path[->] (v53) edge node []  {}  (v55);


\node[state, fill= darkgray , label=170: \small $ p_1 $] (v63)  [above left = 0.2 and  4 of v53] {};
\node[state,fill= darkgray  ,label=180: \small  $ p_2 $ ] (v64) [below left = 5mm and 1.5mm of v63] {};

\path[->] (v63) edge node [label=170:\small$ p_1 $]  {}  (v64);


\node  [above left = 0 and 2.2 of v53]  {\Large$\preceq$};
\node  [below left = 0 and 2.05 of v53]  {\color{blue}\Large$\not\preceq_r$};


\node  [above = 0 of v51]  {\Large$s'$};
\node  [above = 0.3 of v63]  {\Large$s$};

\end{tikzpicture}
	\end{center}\vspace*{-4mm}
	\caption{We have that $s\preceq  s'$, but $s\not\preceq_r  s'$ because the marking of the root of $s'$ is too small.}
	\label{fig:order example}
\end{figure}

While this is irrelevant for the results presented here, let us mention
that checking whether $s\preceq s'$ can be done in polynomial time
by adapting a standard algorithm for the subtree problem
(see for instance~\cite{Stadel78}).

\begin{restatable}{lemma}{partialorder}
	The relations $\preceq $ and $\preceq_r$ are quasi-orders.
\end{restatable}

\begin{proof}
	Let, $ s,s',s''$ be states of an RPN $ \N $ with $s=(r,m_0,\{(m_i,s_i)\}_{1\leq i\leq k})$,
	$s'=(r',m'_0,\{(m'_i,s'_i)\}_{1\leq i\leq k'})$ and  $s''=(r'',m''_0,\{(m''_i,s''_i)\}_{1\leq i\leq k''})$.
	Let us show that the relation $\preceq $ is a quasi-order.
	\begin{enumerate}
		\itemsep=0.9pt
		\item Reflexivity: the identity function $ Id$ on  $V_s$ insures that $s\preceq s$.
		\item Transitivity: Given $ s\preceq s'\preceq s'' $,
		there exist two injective functions  $f:V_s \rightarrow V_{s'}$ and  $f':V_{s'} \rightarrow V_{s''}$.
		Let $ g: V_{s} \rightarrow V_{s''}$ be defined by $ g=f'\circ f $. Then $g$ is injective.
		For any edge $v\xrightarrow{m}_sw$, there exists an edge $f(v) \xrightarrow{m'}_{s'}f(w)$ with $m\leq m'$
		and there exists an edge $f'(f(v)) \xrightarrow{m''}_{s''}f'(f(w))$ with $m\leq m'\leq m''$.
		For all $v\in V_s$, one has
		$M_s(v)  \leq M_{s'}(f(v)) \leq M_{s''}(f'(f(v)))=M_{s''}(g(v)).$
		Therefore  $s \preceq s''$.	
	\end{enumerate}
	The proof for the relation $\preceq_r$ is similar.
\end{proof}
Consider the equivalence relation $\simeq:=\preceq \cap \preceq^{-1}$.
Given a set of states $A$, one denotes by $\bigslant{A}{\simeq}$ the quotient set
by the equivalence relation $\simeq$. Observe that $s \simeq s'$ if and only if their
abstract representations are equal and that $\simeq=\preceq_r \cap \preceq^{-1}_r$.


A quasi-order $\leq$ on the states of an RPN
is \emph{strongly compatible} (as in \cite{Alain01}) if for all states $s,s'$ such that $s\leq s'$
and for all transition firings  $s\xrightarrow{(v,t)}s_1$, there exist a state  $s'_1$ and a transition firing
$s'\xrightarrow{(v',t')}s'_1$ with $s_1 \leq s'_1$.
\begin{restatable}{lemma}{strongcompatibility}
	\label{lem:strongcompatibility}
	The quasi-orders $\preceq $ and  $\preceq_r$ are strongly compatible.
\end{restatable}
\begin{proof}
	Let $s\preceq s'$ and let $f$ be the mapping associated with the relation $\preceq$ and $s\xrightarrow{v,t}s_1$.\\
	Thus $s_v\xrightarrow{v,t}s_2$ for some $s_2$.\\
	We will exhibit some $s'_1$ such that $s_1\preceq s'_1$ with some $f'$ as associated mapping.\\
	Since $M_s(v) \leq M_{s'}(f(v))$, one has $s'_{f(v)}\xrightarrow{f(v),t}s'_2$ for some $s'_2$ and by induction $s'\xrightarrow{f(v),t}s'_1$
	for some $s'_1$. \medskip\\
	It remains to define $f'$.
	\begin{itemize}[nosep]
		\item If $t \in T_{el}$ then $f'=f$;
		\item If $t \in T_{ab}$ then for all threads $u$ of $s$, $f'(u)=f(u)$  and if $v^*$ (resp. $w^*$)
		is the thread created by the firing $(v,t)$ (resp. $(f(v),t)$) then $f(v^*)=w^*$;
		\item If $t \in T_{\tau}$ then $f'$ is equal to $f$ restricted to the remaining vertices.
	\end{itemize}
	It is routine to check that the inequalities between corresponding markings of $s$ and $s'$ are fulfilled.
	The proof for $\preceq_r$ is similar.	
\end{proof}

These quasi-orders may contain an infinite set of
incomparable states (i.e. an infinite \emph{antichain}).
For example, see Figure~\ref{fig:antichain} where any two states $s_i$ and $s_j$ are incomparable.

\begin{figure}[h]
	\begin{center}

\begin{tikzpicture}[
			scale=0.7,
            > = stealth, 
            shorten > = 1pt, 
            auto,
            semithick 
        ]
\usetikzlibrary{arrows,shapes,snakes,automata,backgrounds,petri,positioning,decorations.pathreplacing}
        \tikzstyle{state}=[
        	circle,
            draw = black,
            thick,
            fill = white,
            minimum size = 1mm
        ]
        \tikzstyle{thread}=[circle,thick,draw=black!75,minimum size=3mm]
        
        \small

        \path (-5,-1.5) node[label={[xshift=0.0cm, yshift=-0.9cm]\scriptsize{$p_r$}},draw,circle,inner sep=2pt, minimum width= 0.5cm] (proot) {};
        \path (-5,-3.5) node[label={[xshift=0.0cm, yshift=-0.1cm]\scriptsize{$p_{\ell}$}},draw,circle,inner sep=2pt, minimum width= 0.5cm] (pleaf) {};

        \path (-6,0) node[label={[xshift=0.0cm, yshift=0cm]\scriptsize{$\tau_{r}$}},
        fill=black, draw,rectangle,inner sep=2pt, minimum width= 0.5cm] (tauroot) {};
        \path (-4,0) node[label={[xshift=0.0cm, yshift=0cm]\scriptsize{$t_{r}$}},
        double,draw,rectangle,inner sep=2pt, minimum width= 0.5cm] (troot) {};

        \path (-6,-5) node[label={[xshift=0.0cm, yshift=-0.6cm]\scriptsize{$\tau_{\ell}$}},
        fill=black, draw,rectangle,inner sep=2pt, minimum width= 0.5cm] (tauleaf) {};
        \path (-4,-5) node[label={[xshift=0.0cm, yshift=-0.6cm]\scriptsize{$t_{\ell}$}},
        double,draw,rectangle,inner sep=2pt, minimum width= 0.5cm] (tleaf) {};

       \draw[arrows=-latex'] (proot)-- (tauroot);
       \draw[arrows=-latex'] (proot)-- (-4,-0.75) -- (troot);
       \draw[arrows=-latex'] (troot)--  (-5,-0.75)--(proot);

       \draw[arrows=-latex'] (pleaf)-- (tauleaf);
       \draw[arrows=-latex'] (pleaf)-- (-4,-4.25) -- (tleaf);
       \draw[arrows=-latex'] (tleaf)--  (-5,-4.25)--(pleaf);

       \node (s1) at (-1,0) { \large{$s_1$}};
       \node (s2) at (-1,-2.5) { \large{$s_2$}};
	\node (sn) at (-1,-5) { \large{$s_n$}};
	\node at (-1,-3.75) {\Large{$\vdots$}};
        
        \node[state,  label=270: \small \color{darkgray}$\bf 0$, minimum width= 0.9cm] (v11)  [right = of s1]  {\tiny{$v_0$}};
        \node[state, label=270: \small \color{darkgray}$\bf 0$,  minimum width= 0.9cm] (v12) [right = of v11] {\tiny{$v_1$}};
        \node[state, label=270: \small \color{darkgray}$p_{\ell}$,  minimum width= 0.9cm] (v13) [right =  of v12] {\tiny{$v_2$}};
        
        \path[->] (v11) edge node {$p_r$} (v12);
        \path[->] (v12) edge node {$p_{\ell}$} (v13);

        \node[state,  label=270: \small \color{darkgray}$\bf 0$,  minimum width= 0.9cm] (v21) [right = of s2]   {\tiny{$v_0$}};
        \node[state, label=270: \small \color{darkgray}$\bf 0$,  minimum width= 0.9cm] (v22) [right = of v21]  {\tiny{$v_1$}};
        \node[state, label=270: \small \color{darkgray}$\bf 0$,  minimum width= 0.9cm] (v23) [right = of v22] {\tiny{$v_2$}};-5,2.5
        \node[state,  label=270: \small \color{darkgray}$p_{\ell}$,  minimum width= 0.9cm] (v24) [right = of v23]{\tiny{$v_3$}};

        \path[->] (v21) edge node {$p_r$} (v22);
        \path[->] (v22) edge node {$p_{\ell}$} (v23);
        \path[->] (v23) edge node {$p_{\ell}$} (v24);

        \node[thread,  label=270: \small \color{darkgray}$\bf 0$] (v31) [right = of sn] {\small{$v_0$}};
        \node[state, label=270: \small \color{darkgray}$\bf 0 $,  minimum width= 0.9cm] (v33) [right = of v31] {\tiny{$v_1$}};
        \node[state,  label=270: \small \color{darkgray}$\bf 0 $,  minimum width= 0.9cm] (v35) [right = of v33] {\tiny{$v_2$}};
        \node[state,  label=270: \small \color{darkgray}$\bf 0 $,  minimum width= 0.9cm] (v36) [right = of v35] {\tiny{$v_3$}};
        \node(v39) [right = of v36] {};
        \node[state,  label=270: \small \color{darkgray}$p_{\ell}$,  minimum width= 0.9cm](v41) [right = of v39] {\tiny{$v_{n+1}$}};
        
        \path[->] (v31) edge node {$p_{r}$} (v33);
        \path[->] (v33) edge node {$p_{\ell}$}  (v35);
        \path[->] (v35) edge node {$p_{\ell}$}  (v36);
	 	\draw[](v36) edge  ($(v36)!8mm!(v39)$) edge [dotted] ($(v36)!20mm!(v39)$) ;
        \path[->] (v39) edge node{$p_{\ell}$}  (v41);

\end{tikzpicture}
	\end{center}\vspace*{-5mm}
	\caption{An RPN with an antichain of states}
	\label{fig:antichain}\vspace*{-2mm}
\end{figure}
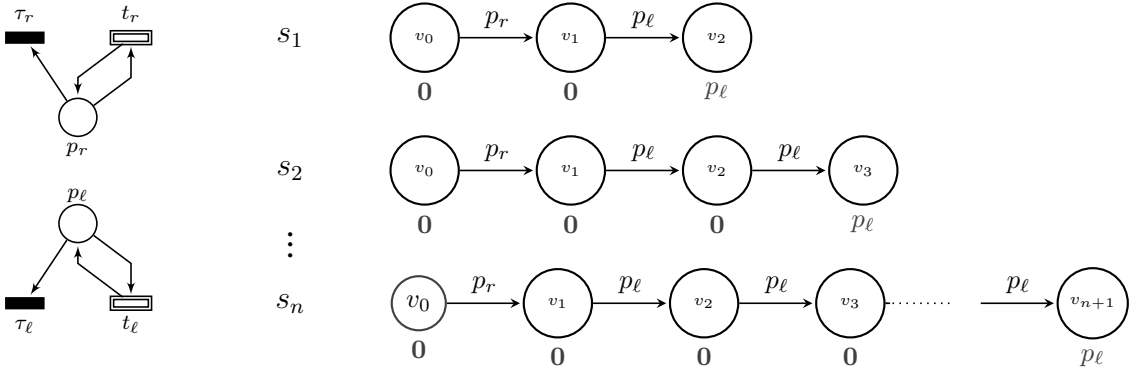

\medskip
Indeed, for any $ i<j $: (1)  $  s_j\not\preceq s_i$
since $ |V_{s_j}|>|V_{s_{i}}| $ there cannot be any injective function from $ V_{s_j}$ to $ V_{s_{i}} $, and (2) $ s_i\not\preceq s_j$
since for any injective function from $ V_{s_i}$ to $ V_{s_{j}} $,
at least one of the edges with the marking $ p_r $ would be mapped to an edge with a marking $p_\ell$.
Since $s\preceq_r s'$ implies $s\preceq s'$, this is also an antichain for $\preceq_r$.

\eject  
Observe also that these quasi-orders are not only strongly compatible. They are \emph{transition-pre\-serving compatible}
meaning that  for all states $s,s'$ such that $s\leq s'$
and for all transition firings  $s\xrightarrow{(v,t)}s_1$, there exist $s'_1$ and a transition firing
$s'\xrightarrow{(v',t)}s'_1$ with $s_1 \leq s'_1$. In Petri net, the standard order on $\nat^P$ is
a well quasi-order which is transition-preserving compatible. The next proposition establishes
that such a quasi-order does not exist in RPN.

\begin{proposition} There does not exist a well quasi-order on states of RPN
	which is transition-pre\-serving compatible.
\end{proposition}
\begin{proof}
	Consider the net of Figure~\ref{fig:antichain} and the family of states $\{s_n\}_{n\geq 1}$.
	By a simple examination one gets that for all $n\geq 1$,
	$s_n\xrightarrow{(v_{n+1},\tau_{\ell})\ldots (v_{1},\tau_{\ell})(v_{0},\tau_{r})} \emptyset$.
	Moreover for all $n'\neq n$, there does not exist a firing sequence from $s_{n'}$ labelled
	by $\tau_{\ell}^{n+1}\tau_{r}$. Thus for any transition-preserving compatible quasi-order $\leq$,
	these states are incomparable establishing that $\leq$ is not a well quasi-order.	
\end{proof}

Since $\preceq$ is not a well quasi-order, RPNs with the relation $\preceq$  are not well structured transition
systems (WSTS)~\cite{Alain01} for which coverability is decidable. 
Therefore to solve coverability, one needs to find another way.

\section{Decision problems and reductions}
\label{sec:reductions}
In this section, we introduce the decision problems that we are going to solve and establish reductions to simpler problems in order to shorten the proofs of subsequent sections.

\medskip
Let $(\N,s_0)$ be a marked RPN and $s_f$ be a state of $\N$.
\begin{itemize}
	\itemsep=0.95pt
	\item The \emph{cut problem} asks whether there exists a firing sequence $\sigma$ such that $s_0\xrightarrow{\sigma}\emptyset$?
	\item The \emph{coverability problem} asks whether there exists a firing sequence $\sigma$ such that $s_0\xrightarrow{\sigma}s\succeq s_f$?
	\item The \emph{termination problem} asks whether there exists an infinite firing sequence?
	\item The \emph{finiteness problem} asks whether $Reach(\N,s_0)$ is finite?
	\item The \emph{boundedness problem} asks whether there exists $B\in \nat$ such that for all $s\in Reach(\N,s_0)$
	and for all $v\in V_s$, one has $\max(M_s(v)(p))_{p\in P} \leq B$?
\end{itemize}
Observe that contrary to Petri nets, the finiteness and boundedness problems are different and not equivalent.
Indeed, an RPN can be bounded while due to the unbounded number of vertices, its reachability set can be infinite.

We introduce the ``rooted'' version of the above problems: for these versions, $s_0$ is required to be some $s[r,m_0]$.
In order to establish a reduction from the general problems to their rooted versions,
given a marked RPN $(\N,s_0)$, we build a marked RPN $(\rooted{\N},s[r,\rooted{m_0}])$ that in a way simulates
the former marked RPN.
We do this by adding a place $p_v$ for every vertex $v\neq r$ of $s_0$
and we add an abstract transition $t_v$ that consumes a token from this place and creates a new vertex with initial marking in $M_{s_0}(v)+\sum_{v\xrightarrow{m_{v'}}_{s_0} v'}$. This will allow to create the children of $v$ in $s_0$ (see Figure \ref{fig:RPN_to_rooted_rpn}).
In order to similarly proceed in the root, $\rooted{m_0}=M_{s_0}(r)+\sum_{r\xrightarrow{m_{v'}}_{s_0} v'}$.

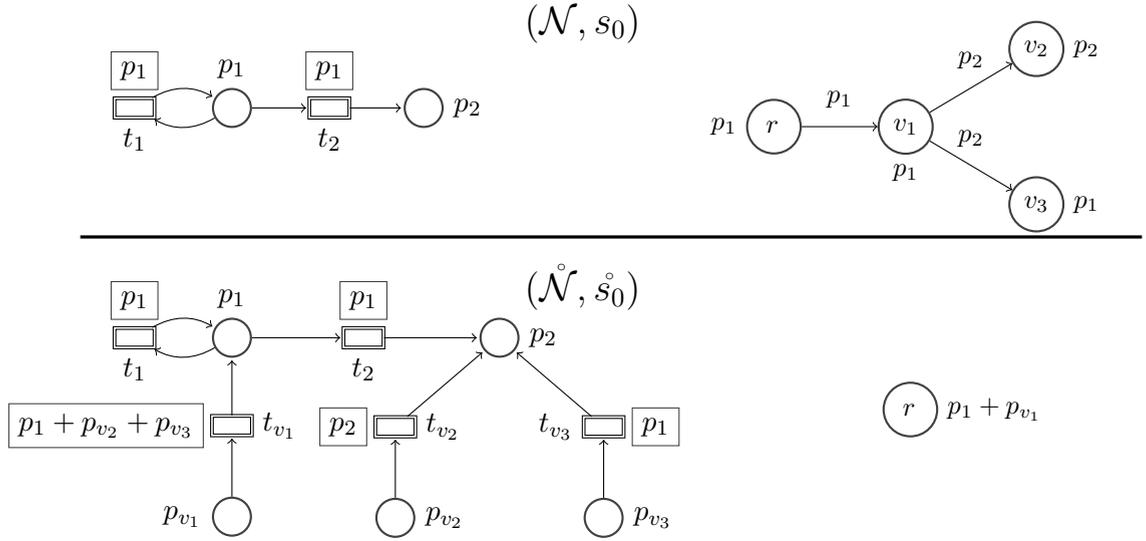
\begin{figure}[h]
\begin{center}
	\begin{tikzpicture}[]
		
		\tikzstyle{place}=[circle,thick,draw=black!75,minimum size=5mm]
		\tikzstyle{thread}=[circle,thick,draw=black!75,minimum size=3mm]
		\tikzstyle{elementary_transition}=[rectangle,minimum width=15, minimum height=1,draw=black!75,]
		\tikzstyle{abstract_transition}=[rectangle,minimum width=15, minimum height=0.5,double,draw=black!100,]
		\tikzstyle{abstract_transition_description}=[rectangle,minimum width=5, minimum height=5,draw=black!75,]
		\usetikzlibrary{arrows,shapes,snakes,automata,backgrounds,petri}
		\tikzstyle{cut_transition}=[rectangle,minimum width=15, minimum height=0.5,fill=black,draw=black!100,]
		
		\node [place, tokens=0,label=90:$p_{1}$] (p1)  {};
		\node [place, tokens=0,label=0:$p_{2}$] (p2)  [right = 2 of p1]  {};
		
		\node [abstract_transition,label=270:$t_{2}$] (t1) [right = 0.75  of p1] {}
		edge[pre] (p1)
		edge[post] (p2);
		\node [abstract_transition_description] [above = 1mm of t1]  {$p_1$};
		
		\node [abstract_transition,label=270:$t_{1}$] (t1) [left = 0.75  of p1] {}
		edge[pre,bend right] (p1)
		edge[post,bend left] (p1);
		\node [abstract_transition_description] [above = 1mm of t1]  {$p_1$};
		
		\node [place, tokens=0,label=90:$p_{1}$] (p21) [below = 2.5 of p1]   {};
		\node [place, tokens=0,label=0:$p_{2}$] (p22)  [right = 3 of p21]  {};
		\node [place, tokens=0,label=180:$p_{v_1}$] (pv1)  [below = 1.83 of p21]  {};
		\node [place, tokens=0,label=1800:$p_{v_2}$] (pv2)  [below left = 2 and 1 of p22]  {};
		\node [place, tokens=0,label=0:$p_{v_3}$] (pv3)  [below right = 2 and 1 of p22]  {};
		
		\node [abstract_transition,label=270:$t_{1}$] (t1) [left = 0.75  of p21] {}
		edge[pre,bend right] (p21)
		edge[post,bend left] (p21);
		\node [abstract_transition_description] [above = 1mm of t1]  {$p_1$};
		\node [abstract_transition,label=270:$t_{2}$] (t1) [right = 1.2  of p21] {}
		edge[pre] (p21)		edge[post] (p22);
		\node [abstract_transition_description] [above = 1mm of t1]  {$p_1$};

		\node  [abstract_transition,label=0:$t_{v_1}$] (ta2) [above = 0.8 of pv1] {}
				edge[pre] (pv1)
				edge[post] (p21);
		\node [abstract_transition_description] [left = 1mm of ta2]  {$p_1 + p_{v_2}+p_{v_3}$};
			
		\node  [abstract_transition,label=0:$t_{v_2}$] (ta2) [above = 0.8 of pv2] {}
		edge[pre] (pv2)
		edge[post] (p22);
		\node [abstract_transition_description] [left = 1mm of ta2]  {$p_2$};

		\node  [abstract_transition,label=180:$t_{v_3}$] (ta2) [above = 0.8 of pv3] {}
		edge[pre] (pv3)
		edge[post] (p22);
		\node [abstract_transition_description] [right = 1mm of ta2]  {$p_1 $};

		%
		%
				{\small
					\node[thread, label=180: $ p_1 $ ] (vr) [ below right   = -0.2 and 6.7 of p1]  {$r^{{}^{}}_{{}}$};
					\node[thread, label=270:   $p_1$] (v1) [ right  = 1 of vr] {$ v_1 $};
					\node[thread, label=0:   $p_2$] (v2) [above  right  = 0.5 and 1.2  of v1] {$ v_2 $};
					\node[thread, label=0:   $p_1$] (v3) [below  right  = 0.5 and 1.2  of v1] {$ v_3 $};
					
					\path[->] (vr) edge node[label=90:$p_1$]  {} (v1);
					\path[->] (v1) edge node[label=$p_{2}$]  {} (v2);
					\path[->] (v1) edge node[label=$p_{2}$]  {} (v3);
				}
				{\small
				\node[thread, label=0: $ p_1 + p_{v_1} $ ] (vr) [ below right   = 0.5 and 8.5 of p21]  {$r^{{}^{}}_{{}}$};
			}
		\Large
		\node [above right  = 0.5 and 3.5 of p1] (N) {{$(\N,s_0)$}};
		\node [above right  = -0 and 3.5 of p21] (rN) {{$(\rooted{\N},\sr)$}};
		
		%
		%
				\draw[color=black,very thick](-2,-1.7) -- (12,-1.7);
		%
		%

	\end{tikzpicture}
	
\end{center}\vspace*{-4mm}
	\caption{From a marked RPN to a rooted one}\label{fig:RPN_to_rooted_rpn}\vspace*{-3mm}
\end{figure}

\begin{definition}
	Let $(\N,s_0)$ be a marked RPN. Then $(\rooted{\N},\sr)$ is defined by:
	\begin{itemize}[nosep]
		\item $\rooted{P} = P \cup \{p_v \mid v\in V_{s_0}\setminus\{r_{s_0}\}\}$ ;
		\item  $\rooted{T}_{ab}=T_{ab}\cup T_V$, $\rooted{T}_{\tau} = T_{\tau}$,
		$\rooted{T}_{el}=T_{el}$ with $ T_V=\{t_v \mid v\in  V_s\setminus\{r_s\}\}$ ;
		\item for all $t\in T$,	one has $\rooted{W}^-(t)=W^-(t)$
		and  all $t\in T_{ab}\cup T_{el}$, $\rooted{W}^+(t)=W^+(t)$ ;
		\item for all $t_v\in T_V$ and $u\xrightarrow{m_v}_{s_0} v$, $\rooted{W}^-(t_v)=p_v$ and $\rooted{W}^+(t_v)=m_v$ ;
		\item for all $t\in T_{ab}$, $\rooted{\Omega}(t) = \Omega(t)$ ;
		\item for all $t_v\in T_V$,  $\rooted{\Omega}(t_v)=M_{s_0}(v)+\sum_{v\xrightarrow{m_{v'}}_{s_0} v'} p_{v'}$ ;
		\item $\sr = s[r,M_{s_0}(r_s)+\sum_{r_{s_0}\xrightarrow{m_v}_{s_0} v} p_{v}]$.
	\end{itemize}
	\end{definition}

Let $m\in \nat^{\rooted{P}}$,  we denote by $m_{|_{P}}\in\nat^{P}$ the projection of $m$ on  $P$.
Let $s$ be a state of  $\rooted{N}$,  we denote $s_{|_{P}}$ a state of $\N$ obtained by projecting  every marking of $s$ on $P$.

\medskip
\noindent{\bf Observations.}
\begin{enumerate}[nosep]
	\item The encoding size of $(\N,s_0)$ is linear w.r.t. the encoding size of $(\rooted{\N},\rooted{s_0})$.
	\item Let $e:=(v_i)_{0\leq i\leq k}$ be an enumeration of $V_{s_0}$ such that $v_0=r_{s_0}$
	and for all $0< i\leq k$, $prd(v_i)\in \{v_j\}_{j<i}$. Consider $\sigma^e_{s_0} = (prd(v_i),t_{v_i})_{i=1}^{k}$.
	Such an enumeration is called \emph{consistent}.
	By construction of $\rooted{\N}$, $\sr \xrightarrow{\sigma^e_{s_0}}_{\rooted{\N}} s'_0$ with $s'_{0|P}=s_0$
	and all places of $P_V$ unmarked in $s'_0$.
	\item Let $\sr\xrightarrow{\sigma}_{\rooted{\N}}s$. Then by construction, for all $v\in V_{s_0} \setminus \{r_{s_0}\}$,
	there is at most one occurrence of $t_v$ which furthermore is fired in $prd(v)$. Moreover since these firings consume
	tokens in $P_V$ that were not used for firings of $T$, they can be  pushed at the beginning of $\sigma$ (denoted by $\sigma_1$)
	and completed by the missing firings of $T_V$ in $\sigma$ (denoted by $\sigma_2$) getting a consistent enumeration $e$.
	Summarizing, denoting $\sigma_{| \N}$, $\sigma$ without the firings of $T_V$,
	one gets that: \\
	(1)  $\sr\xrightarrow{\sigma_1\sigma_{| \N}}_{\rooted{\N}}s$,\\
	(2) $\sr\xrightarrow{\sigma\sigma_2}_{\rooted{\N}}s'$ and \\
	(3) $s_0 \xrightarrow{\sigma_{| \N}}_{\N}s''$
	with $s'_{|P}=s''$ and all places of $P_V$ are unmarked in $s'$.
	\end{enumerate}

\smallskip\noindent
Due to observation~2, we immediately get that:
\begin{lemma}
	\label{lemma:fromNtoroot}
	Let $(\N,s_0)$ be a marked RPN and $s_0\xrightarrow{\sigma}_{\N}s$.
	Then  for every consistent enumeration $e$, there exists a firing sequence $\sr\xrightarrow{\sigma^e_{s_0}{\sigma}}_{\rooted{\N}}s'$ with ${s'}_{|_P} = s$
	and all places of $P_V$ are unmarked in $s'$.
\end{lemma}

Due to observation~3, we immediately get that:
\begin{lemma}
	\label{lemma:fromroottoN}
	Let $(\N,s_0)$ be a marked RPN and $\sr\xrightarrow{\sigma}_{\rooted{\N}}s$.
	Then there exist a consistent enumeration $e$ and a decomposition $\sigma^e_{s_0}=\sigma_1\sigma_2$
	such that $\rooted{s_0}\xrightarrow{\sigma_1\sigma_{| \N}}_{\rooted{\N}}s$,
	$\sr\xrightarrow{\sigma\sigma_2}_{\rooted{\N}}s'$ and $s_0 \xrightarrow{\sigma_{| \N}}_{\N}s''$
	with $s'_{|P}=s''$ and all places of $P_V$ are unmarked in $s'$.
\end{lemma}

Due to the previous lemmas, we get that:
\begin{proposition}\label{col:rooted}
	The cut (resp. coverability, termination, finiteness, boundedness) problem
	is polynomially reducible to the rooted cut (resp. coverability, termination, finiteness, boundedness)  problem.
\end{proposition}

\begin{proof}
	\noindent
	Let $(\N,s_0)$ be a marked RPN and $s_f$ be a state of $\N$. Define $\rooted{s}_f$
	a state of $\rooted{\N}$ be as $s_f$ with in all markings of $s_f$, all places of $\rooted{P} \setminus P$ unmarked.
	
	\noindent
	$\bullet$ Assume that there exists  $s_0\xrightarrow{\sigma}_{\N}\emptyset$. Then by Lemma~\ref{lemma:fromNtoroot},
	$\sr\xrightarrow{\sigma^e_{s_0}{\sigma}}_{\rooted{\N}}\emptyset$. Assume that there exists
	$\sr\xrightarrow{\sigma}_{\rooted{\N}}\emptyset$ which means that the last transition is fired in the root
	and is a cut transition. Then by Lemma~\ref{lemma:fromroottoN}, $s_0 \xrightarrow{\sigma_{| \N}}_{\N}s''$ for some $s''$. Since
	the last firing of  of $\sigma_{| \N}$ is the cut transition fired in the root $s''=\emptyset$.
	
	\noindent
	$\bullet$ Assume that there exists  $s_0\xrightarrow{\sigma}_{\N}s\succeq s_f$. Then by Lemma~\ref{lemma:fromNtoroot},
	$\sr\xrightarrow{\sigma^e_{s_0}{\sigma}}_{\rooted{\N}}\rooted{s}$ with $\rooted{s}_{|P}=s$.
	Thus  $\rooted{s}\succeq \rooted{s}_f$. Assume that there exists
	$\sr\xrightarrow{\sigma}_{\rooted{\N}}s\succeq \rooted{s}_f$.
	Then by Lemma~\ref{lemma:fromroottoN}, there exists $\sigma_2$ a firing sequence of $T_V$
	with  $\sr\xrightarrow{\sigma\sigma_2}_{\rooted{\N}}s'$,  $s_0 \xrightarrow{\sigma_{| \N}}_{\N}s''$
	and $s'_{|P}=s''$. Since $\sigma_2$ only creates vertices and deletes tokens from $P_V$,
	$s'\succeq \rooted{s}_f$. Thus $s'' \succeq s_f$.
	
	\smallskip\noindent
	$\bullet$ Assume that there exists  $s_0\xrightarrow{\sigma}_{\N}$ with $\sigma$ infinite. Then by Lemma~\ref{lemma:fromNtoroot},
	$\sr\xrightarrow{\sigma^e_{s_0}{\sigma}}_{\rooted{\N}}$. Assume that there exists
	$\sr\xrightarrow{\sigma}_{\rooted{\N}}$ with $\sigma$ infinite.
	Then by Lemma~\ref{lemma:fromroottoN}, $s_0 \xrightarrow{\sigma_{| \N}}_{\N}$ with $\sigma_{| \N}$ infinite
	since there are only a finite number of firings of $T_V$.
	
	\smallskip\noindent
	$\bullet$ Assume that $Reach(\N,s_0)$ is infinite. For all $s \in Reach(\N,s_0)$, define $\rooted{s}$ a state
	of $\rooted{\N}$  as $s$  with all places of $P_V$ in  markings of $s$ unmarked. Due to
	Lemma~\ref{lemma:fromNtoroot}, $\rooted{s}\in Reach(\rooted{\N},\rooted{s_0})$. Since this mapping is injective,
	$Reach(\rooted{\N},\rooted{s_0})$ is infinite. Assume that $Reach(\rooted{\N},\rooted{s_0})$ is infinite.
	Let $s \in Reach(\rooted{\N},\rooted{s_0})$. Due to  Lemma~\ref{lemma:fromroottoN},
	consider $s \xrightarrow{\sigma_2}_{\rooted{\N}}s'$ and $s_0 \xrightarrow{\sigma_{| \N}}_{\N}s''$
	with $s'_{|P}=s''$ and all places of $P_V$ unmarked in $s'$. Thus $s''\in Reach(\N,s_0)$.
	The mapping from $s$ to $s''$ is not injective. However, the inverse image of $s''$ by this mapping is finite since there are a finite number of consistent enumerations and prefixes of such enumerations. Thus $Reach(\N,s_0)$ is infinite.
	
	\smallskip\noindent
	$\bullet$ Assume that $(\N,s_0)$ is unbounded. For all $s \in Reach(\N,s_0)$, define $\rooted{s}$ a state
	of $\rooted{\N}$  as $s$  with all places of $P_V$ in  markings of $s$ unmarked. Due to
	Lemma~\ref{lemma:fromNtoroot}, $\rooted{s}\in Reach(\rooted{\N},\rooted{s_0})$.
	Thus $(\rooted{\N},\rooted{s_0})$ is unbounded. Assume that $(\rooted{\N},\rooted{s_0})$ is unbounded.
	By construction, the marking of places in $P_V$ is bounded.
	Let $s \in Reach(\rooted{\N},\rooted{s_0})$. Due to  Lemma~\ref{lemma:fromroottoN},
	consider $s \xrightarrow{\sigma_2}_{\rooted{\N}}s'$ and $s_0 \xrightarrow{\sigma_{| \N}}_{\N}s''$
	with $s'_{|P}=s''$ and all places of  $P_V$ unmarked in $s'$. Thus $s''\in Reach(\N,s_0)$.
	Since for all vertex $v$ of $s$, $v$ is also present in $s''$ and for all $p\in P$,
	$M_s(v)(p)=M_{s''}(v)(p)$. Then $Reach(\N,s_0)$ is unbounded.
\end{proof}

Let $\sigma$ be a firing sequence.  A thread is \emph{extremal} w.r.t.  $\sigma$ if it is an initial or final thread.

\begin{definition}
	Let $\N$ be an RPN. Then $T_{ret}\subseteq T_{ab}$, the set of \emph{returning transitions} is defined by:
	$$\{t\in T_{ab}\mid \exists s[r,\Omega(t)] \xrightarrow{\sigma}\emptyset\}$$
\end{definition}
For all $t\in T_{ret}$, we define $\sigma_{t}$ to be some arbitrary
shortest \emph{returning sequence} (i.e. $s[r,\Omega(t)] \xrightarrow{\sigma_t}\emptyset$).
We now introduce $\widehat{\N}$ obtained from $\N$ by adding elementary transitions that mimic
the behaviour of a returning sequence. Observe that the size of $\widehat{\N}$
is linear w.r.t. the size of $\N$.

\begin{definition} Let $\N$ be an RPN. Then
	$\widehat{\N}=\left<P,\widehat{T},\widehat{W}^+,
	\widehat{W}^-,\Omega\right>$ is defined by:
	\begin{itemize}[nosep]\medskip
		\item $\widehat{T}_{ab}=T_{ab}$, $\widehat{T}_{\tau} = T_{\tau} $ ,
		$\widehat{T}_{el}=T_{el}\uplus \{t^r\mid t\in T_{ret} \}$;
		\item for all $t\in T$,	$\widehat{W}^-(t)=W^-(t)$
		and  all $t\in T_{ab}\cup T_{el}$, $\widehat{W}^+(t)=W^+(t)$;
		\item for all $t\in T_{ab}$,	$\widehat{\Omega}(t)=\Omega(t)$;
		\item for all $t\in T_{ret}$,	$\widehat{W}^-(t^r)=W^-(t)$
		and  $\widehat{W}^+(t^r)=W^+(t)$.
	\end{itemize}\medskip\noindent
Figure~\ref{fig:RPN_hat} has an example of an RPN $\N$ and its $\widehat{\N}$.
\end{definition}

\begin{figure}[h]
\begin{center}
	\begin{tikzpicture}[]
		
		\tikzstyle{place}=[circle,thick,draw=black!75,minimum size=5mm]
		\tikzstyle{thread}=[circle,thick,draw=black!75,minimum size=3mm]
		\tikzstyle{elementary_transition}=[rectangle,minimum width=15, minimum height=1,draw=black!75,]
		\tikzstyle{abstract_transition}=[rectangle,minimum width=15, minimum height=0.5,double,draw=black!100,]
		\tikzstyle{abstract_transition_description}=[rectangle,minimum width=5, minimum height=5,draw=black!75,]
		\usetikzlibrary{arrows,shapes,snakes,automata,backgrounds,petri}
		\tikzstyle{cut_transition}=[rectangle,minimum width=15, minimum height=0.5,fill=black,draw=black!100,]
		
		\node [place, tokens=0,label=170:$p_{1}$] (p1)  {};
		\node [place, tokens=0,label=0:$p_{2}$] (p2)  [below = 2 of p1]  {};
		
			\node  [cut_transition,label=0:$t_{\tau}$] (tau) [above = 0.5 of p1] {}
			edge[pre] (p1);
			
			\node  [abstract_transition,label=0:$t_{1}$] (ta2) [below  =   0.8 of p1] {}
					edge[pre] (p1)
					edge[post] (p2);
			\node [abstract_transition_description] [left = 1mm of ta2]  {$p_1$};
			
			\node  [abstract_transition,label=180:$t_{2}$] (ta2) [below right = 0.9 and 1.2 of p1] {}
			edge[pre] (p1)
			edge[post] (p2);
			\node [abstract_transition_description] [right = 1mm of ta2]  {$p_2$};
		
		\node [place, tokens=0,label=170:$p_{1}$] (p21) [right  = 5 of p1]  {};
		\node [place, tokens=0,label=0:$p_{2}$] (p22)  [below = 2 of p21]  {};
		
		\node  [cut_transition,label=0:$t_{\tau}$] (tau2) [above = 0.5 of p21] {}
		edge[pre] (p21);
		
		\node  [abstract_transition,label=0:$t_{1}$] (ta22) [below  =   0.8 of p21] {}
		edge[pre] (p21)
		edge[post] (p22);
		\node [abstract_transition_description] [left = 1mm of ta22]  {$p_1$};
		
		\node  [abstract_transition,label=180:$t_{2}$] (ta21) [below right = 0.9 and 1.2 of p21] {}
		edge[pre] (p21)
		edge[post] (p22);
		\node [abstract_transition_description] [right = 1mm of ta21]  {$p_2$};
		
		\node [elementary_transition,label=100:$t_{1}^r$] (t21) [left = 1  of ta22] {}
				edge[pre,bend left] (p21)		
				edge[post,bend right] (p22);

		\node [place, tokens=0,label=170:$p_{1}$] (p21) [right  = 5 of p21]  {};
		\node [place, tokens=0,label=0:$p_{2}$] (p22)  [below = 2 of p21]  {};
		
		
		\node  [elementary_transition,label=0:$t_{1}$] (ta22) [below  =   0.8 of p21] {}
		edge[pre] (p21);
		
		\node  [elementary_transition,label=180:$t_{2}$] (ta21) [below right = 0.9 and 1.2 of p21] {}
		edge[pre] (p21);
		
		\node [elementary_transition,label=100:$t_{1}^r$] (t21) [left = 1  of ta22] {}
		edge[pre,bend left] (p21)		
		edge[post,bend right] (p22);

%
%
%
%
%

		%
		%
		\node[above = 1  of p1 ] (N) {{$\N$}};
		\node[right = 4.9 of N ] (Nh) {{$\widehat{\N}$}};
		\node[right = 10.5 of N ] (Nhel) {{$\widehat{\N}_{el}$}};
		
	\end{tikzpicture}
	
\end{center}\vspace*{-3mm}
	\caption{From $\N$ to $\widehat{\N}$ and $\widehat{\N}_{el}$}\label{fig:RPN_hat}\vspace*{-2mm}
\end{figure}

Note that since $\widehat{\N}$ enlarges $\N$ by adding transitions and that any firing of $t^r$ in $\widehat{\N}$
can be replaced by the firing of $t\sigma_t$ in $\N$ we get:
\begin{proposition}
	\label{prop:equivreach}
	Let $(\N,s_0)$ be a marked RPN. Then $Reach(\N,s_0)=Reach(\widehat{\N},s_0)$.
\end{proposition}

We call a firing sequence $\sigma$ \emph{omniscient} if any thread created during its firing is a final thread.

\begin{proposition}\label{prop:omniciant}
	Let $(\N,s_0)$ be a marked RPN and $s_0\xrightarrow{\sigma}_{\N}s$.
	Then there exists  a firing sequence $s_0\xrightarrow{\widehat{\sigma}}_{\widehat{\N}} s$ such that $\widehat{\sigma}$ is omniscient.
\end{proposition}

\begin{proof}
	Assume that we have an extremal thread $u$ which fires $ t\in T_{ab} $ creating a non final thread $ v $ that disappears by a matching cut transition $ (v,t_\tau)\in \sigma $ for $t_\tau\in T_\tau$.
	One builds $\sigma'$ by (1) deleting from $\sigma $ the transition $ (u,t) $, (2) deleting all the firings from $ Des_{\sigma}(v) $ in $\sigma$ and (3) replacing the transition $ (v,t_\tau) $ by $(u,t^r)$. We claim that $s\xrightarrow{\sigma}s' $. Indeed
	in $u$ the transition $(u,t^r)$ has the same incidence in $u$
	as the transition  $(u,t)$ followed by $(v,t_\tau)$ (`anticipating' $(v,t_\tau)$ only add tokens in intermediate states)
	and the other deleted firings are performed by threads in $Des_{\sigma}(v) $ which do not exist anymore.
	By taking $\widehat\sigma$ the sequence obtained by iterating the process, we get the omniscient sequence.
\end{proof}

\noindent
In order to recover from a sequence in $\widehat{\N}$ a sequence
in $\N$,  for every $ t\in T_{ret}$ one has to simulate the firings of a transition
$t^r$ by sequence $\sigma_t$. Therefore bounding the length of $\sigma_t$ is a critical issue.
Recall that in~\cite{Rac78}, Rackoff showed that the coverability  problem for Petri nets
belongs to $\EXPSPACE$.
More precisely, he proved that if there exists a covering sequence,
then there exists a `short' one:
\begin{theorem}[Rackoff \cite{Rac78}]
	\label{thm:Rackoff covering path}Let
	$\N$  be a Petri net,  $m_{ini}$,
	$m_{tar}$ be markings and $\sigma$ be a firing sequence such that
	$m_{ini}\xrightarrow{\sigma}m\geq m_{tar}$.
	Then there exists a sequence $\sigma'$
	such that $m_{ini}\xrightarrow{\sigma'}m'\geq m_{tar}$ with
	$\left|\sigma'\right|\leq 2^{2^{cn\log n}}$
	for some constant $c$ and $n$ being the size of $(\N,m_{tar})$.
	
\end{theorem}
A surprising consequence of Rackoff's proof is that the length of the minimal coverability sequence does not depend on the initial marking of the net.

\begin{proposition}
	\label{prop:Length of an abstract transition in RPN}
	Let $\N$ be an RPN and $t\in T_{ret}$.
	Then the returning sequence $\sigma_t$
	fulfills $|\sigma_t|\leq 2^{\cdot2^{dn\log n}}$ for some constant $d$ and $n=size(\N)$.
\end{proposition}

\begin{proof}
	Let us enumerate
	$T_{ret}=\{t_1,\ldots,t_K\}$ in such a way that $i<j$ implies $|\sigma_{t_i}|\leq|\sigma_{t_j}|$.
	Observe first that the shortest returning sequences do not include
	firings of abstract transitions not followed by a matching cut transition
	since it could be omitted as it only deletes tokens in the thread.
	We argue by induction on $k\leq K$ that:
	\[
	|\sigma_{t_{k}}|<2^{k\cdot2^{cn\log n}} \qquad \mbox{where } c \mbox{ is the Rackoff constant}
	\]
	For $k=1$, we know that $\sigma_{t_{1}}$ has a minimal length over all
	returning sequences. Hence
	there are no cuts
	in $\sigma_{t_{1}}$ except the last one. Due to the above observation,
	$\sigma_{t_{1}}$ only includes firing of elementary transitions.
	Thus the Rackoff bound of Theorem~\ref{thm:Rackoff covering path} applies
	for a covering of some final marking.
	
	\noindent
	Assume that the result holds for all $i<k$.
	Due to the requirement on lengths,  $\sigma_{t_{k}}$
	only includes cuts from threads created by $t_{i}\in T_{ret}$ with $i<k$.
	Thus by Proposition~\ref{prop:omniciant}
	we get a sequence $\widehat{\sigma}_{t_{k}}\cdot(r,t_\tau)$
	in $\widehat{\N}$ (where $r$ is the root and $ t_\tau\in T_\tau $). The sequence
	$ \widehat{\sigma}_{t_{k}} $ consists of only elementary transitions
	and does not contain any transition  $t_i^r$ with $i\geq k$.
	The marking of $r$ reached by $ \widehat{\sigma}_{t_{k}} $ covers
	some final marking, hence	 by Theorem~\ref{thm:Rackoff covering path}
	there exists a covering sequence $\widehat{\sigma}_{t_{k}}'$ such that $|\widehat{\sigma}_{t_{k}}'|\leq2^{2^{cn\log n}}$.
	Since $\widehat{\sigma}_{t_{k}}$ does not contain firing of $t_i^r$ with $i\geq k$
	this also holds for $\widehat{\sigma}_{t_{k}}'$. Substituting any firing of $t_i^r$
	by 	$\sigma_{t_i}$, one gets a corresponding sequence $\sigma_{t_k}'$	in $\N$. Using the induction hypothesis, one gets that
	the length of $\sigma_{t_k}'$ fulfills:
	\[
	|\sigma_{t_k}'|\leq|\widehat{\sigma}_{t_{k}'}|2^{(k-1)\cdot2^{cn\log n}}\leq2^{2^{cn\log n}}\cdot 2^{(k-1)\cdot2^{cn\log n}}\leq2^{k\cdot2^{cn\log n}}
	\]
	From minimality of $\sigma_{t_{k}}$,
	one gets $|\sigma_{t_{k}}|\leq|\sigma_{t_k}'|\leq2^{k\cdot2^{cn\log n}}$
	which concludes the proof since
	\[
	\max_{t\in T_{ret}}\{ |\sigma_{t}|\}\leq2^{|T_{ret}|
		\cdot2^{cn\log n}}\leq2^{n2^{cn\log n}}\leq2^{2^{2cn\log n}}.
	\]

\vspace*{-7mm}
\end{proof}

Using the previous proposition, we can compute $T_{ret}$ in exponential space, by enumerating for all
abstract transitions, all firing sequences of sufficient length and checking whether they lead to the empty tree.

Below are immediate corollaries from the previous  propositions:
\begin{corollary}\label{col:bound_on_translation_of_N-hat_to_N}	
	Let $\N$ be a marked RPN.
	Then for all $s\xrightarrow{\widehat{\sigma}}_{\widehat{\N}}s'$, there exists $s\xrightarrow{\sigma}_{\N}s'$
	such that $|\sigma|\leq 2^{\cdot2^{dn\log n}} |\widehat{\sigma}|$ for some constant $d$ and $n=size(\N)$.
\end{corollary}

\begin{corollary}\label{col:build_N_hat}
	Given an RPN $\N$ one can build $\widehat{\N}$ in exponential space.
\end{corollary}

In order to mimic the behavior of a  specific thread in a firing sequence (which will be useful later on),
we introduce the Petri net $\widehat{\N}_{el}$. The size of $\widehat{\N}_{el}$ is also linear w.r.t. the size of $\N$.

\begin{definition}\label{def:N_el}
	Let $\N$ be an RPN. Then the Petri net $\widehat{\N}_{el}=\left<P,\widehat{T}_{el},\widehat{W}^+_{el},\widehat{W}^-_{el}\right>$  is defined by:
	\begin{itemize}
				\item $\widehat{T}_{el}=\widehat{T}\setminus T_\tau$;
		\eject
		\item For all $t\in \widehat{T}_{el}\setminus T_{ab}$,  $\widehat{W}^{-}_{el}(t)=\widehat{W}^{-}(t)$
		and  $\widehat{W}^{+}_{el}(t)=\widehat{W}^{+}(t)$;  \vspace*{-1mm}
		\item For all $t\in T_{ab}$,  $\widehat{W}^{-}_{el}(t)=\widehat{W}^{-}(t)$ and  $\widehat{W}^+_{el}(t)=0$.
			\end{itemize}
Figure~\ref{fig:RPN_hat} has an example of an RPN $\N$ and its $\widehat{\N}_{el}$.
\end{definition}
As for  $\widehat{\N}$, one can build $\widehat{\N}_{el}$ in exponential space.

\medskip \noindent{\bf Observation.} The main (straightforward) property of $\widehat{\N}_{el}$ is the following one.
Let $\sigma\in \widehat{T}_{el}^*$ with $n_t$ the number of occurrences of $t$ in $\sigma$.
Then $m_0\xrightarrow{\sigma}_{\widehat{\N}_{el}} m$
if and only if $s[r,m_0]\xrightarrow{(r,\sigma)}_{\widehat{\N}} s$ with $V_s=\{r\}\cup \bigcup_{t\in T_{ab}} \{v_{t,1},\ldots,v_{t,n_t}\}$, $M_s(r)=m$
and for all $v_{t_i}$,  $r \xrightarrow{W^+(t)}_s v_{t_i}$  and $M_s(v_{t_i})=\Omega(t)$.

\section{Expressiveness}
\label{sec:expressiveness}

The expressiveness of a formalism may be defined by the family
of languages that it can generate. In~\cite{HP-icatpn99}, the expressiveness
of RPNs was studied using reachability languages. However, using reachability languages as specification languages has an
inconvenient since the emptiness problem for these languages is not elementary \cite{DBLP:conf/stoc/CzerwinskiLLLM19} for Petri nets,
so it is also not elementary, at least, for RPN. We propose to characterize the expressive
power of RPN by studying the family of coverability languages
which is sufficient to express most of the usual reachability properties since many of them reduce to check that no reachable state
may cover a bad marking in a thread.

%

The characterization of the expressive power by means of covering languages has been done for Petri nets (studied in the book of Peterson  \cite{nla.cat-vn2956435}), and more recently, for Well Structured Transition Systems (WSTS) \cite{DBLP:journals/acta/GeeraertsRB07} and for monotonic extensions of Petri nets like reset-transfer Petri nets, $\nu$-Petri nets, unordered Petri nets \cite{BFHR-icomp13,DBLP:journals/tcs/DelzannoR13}.
More properties are decidable for VASS covering languages than for VASS reachability languages.
For instance, universality for reachability languages is undecidable for 1-VASS \cite{DBLP:journals/jcss/ValkV81} and then co-finiteness is also undecidable but these two properties are both decidable for VASS covering languages \cite{figueira:hal-02193089}; moreover, it is Ackermann-complete for 1-VASSs \cite{DBLP:conf/rp/HofmanT14}. Generally, the universality of both reachability and coverability of WSTS languages is undecidable \cite{DBLP:journals/acta/GeeraertsRB07}.

So we equip any transition $t$ with a \emph{label} $\lambda(t)\in \Sigma \cup \{\varepsilon\}$
where $\Sigma$ is a finite alphabet and $\varepsilon$ is the empty word.
The labelling is extended to transition sequences in the usual way.
Thus given a labelled marked RPN $(\N,s_0)$ and a finite subset of states $S_f$,
the (coverability) language $\mathcal L_C(\N,s_0,S_f)$ is defined by:
$$\mathcal L_C(\N,s_0,S_f)=\{\lambda(\sigma) \mid \exists\ s_0 \xrightarrow{\sigma} s\succeq s_f\wedge s_f \in S_f\}$$
i.e. the set of labellings for sequences covering some state of $S_f$ in $\N$.

\medskip
We now study the family of RPN coverability languages both from the point of view of expressiveness and closure under multiple operations.

\begin{restatable}{proposition}{closedbyunion}
	\label{prop:closedunion}
	The family of coverability languages of RPNs is closed under union.
\end{restatable}

\begin{proof}
	We closely follow the classic proof that the family of Petri net languages is closed under union, i.e. adding a place and two extra transitions that have to be fired in the beginning of the firing sequence in

\eject

\noindent order to decide in which of the Petri net one fires. Due to the correspondence between firing sequences of $(\N,s_0)$ and those of $(\rooted{\N},\sr)$,
	established in the previous section,
	one can assume w.l.o.g. that the initial markings of the RPNs have a single vertex.
	Consider two labelled marked RPNs with final states
	$(\N,s[r,m_0],S_f)$ and $(\N',s[r',m'_0],S'_f)$.
	Let us define $\widetilde{\N}$ as follows.
	Its set of places is the disjoint union of $P$ and $P'$
	with three additional places $p_0$, $p$ and $p'$.
	Its set of transitions is the disjoint union of $T$ and $T'$
	with four additional elementary transitions $t_b$, $t_c$, $t'_b$ and $t'_c$.
	
	\medskip\noindent
	$\bullet$ For all $t\in T$, $\widetilde{W}^-(t)=W^-(t)+p$ and when $t\notin T_{\tau}$
	$\widetilde{W}^+(t)=W^+(t)$\vspace{0.8mm}\\
	$\bullet$ For all $t\in T'$, $\widetilde{W}^-(t)=W'^-(t)+
	p'$ and when $t\notin T'_{\tau}$ $\widetilde{W}^+(t)=W'^+(t)$\vspace{0.8mm}\\
	$\bullet$ For all $t\in T_{ab}$, $\widetilde{\Omega}(t)=\Omega(t)+ p$\vspace{0.8mm}\\
	$\bullet$ For all $t\in T'_{ab}$, $\widetilde{\Omega}(t)=\Omega'(t)+ p'$\vspace{0.8mm}\\
	$\bullet$ $\widetilde{W}^-(t_b)=\widetilde{W}^-(t'_b)= p_0$,
	$\widetilde{W}^+(t_b)=m_0+ p$, $\widetilde{W}^+(t'_b)=m_0'+ p'$\vspace{0.8mm}\\
	$\bullet$ $\widetilde{W}^-(t_c)=p$, $\widetilde{W}^+(t_c)=2 p$,
	$\widetilde{W}^-(t'_c)= p'$, $\widetilde{W}^+(t'_c)=2 p'$\vspace{0.8mm}\\
	$\bullet$ $\widetilde{S}_f$ is obtained from the union
	$S_f\cup S'_f$ by adding a token in place
	$p$ (resp. $p'$)\vspace{0.8mm}\\
	\hspace*{0.3cm}of all markings of states of $S_f$ (respectively $S'_f$).\vspace{0.8mm}\\
	$\bullet$ For all $t\in T$, $\widetilde{\lambda}(t)=\lambda(t)$ and
	for all $t\in T'$, $\widetilde{\lambda}(t)=\lambda'(t)$\vspace{0.8mm}\\
	$\bullet$ For all $t\in \{t_b,t_c,t'_b,t'_c\}$,
	$\widetilde{\lambda}(t)=\varepsilon$.\vspace{0.8mm}\\
	$\bullet$ The initial state of $\widetilde{\N}$ is $s[\tilde{r},p_0]$.
	
	\medskip\noindent
	Let us prove that $\mathcal L(\N,s[r,m_0],S_f) \cup \mathcal L(\N',s[r',m'_0],S'_f)
	\subseteq \mathcal L(\widetilde{\N},s[\tilde{r},p_0],\widetilde{S}_f)$.
	Let $\sigma$ be a coverability sequence of $(\N,s[r,m_0],S_f)$.
	The corresponding coverability sequence $\widetilde{\sigma}$
	of $L(\widetilde{\N},s[\tilde{r},p_0],\widetilde{S}_f)$ is built as follows.
	Initially,
	one fires $(\tilde{r},t_b)(\tilde{r},t_c)^{\ell_r}$ where $\ell_r$ is the number of abstract transition firings
	occurring in $\sigma$ triggered by $r$. Then after the creation
	of a thread $v$, one inserts $(v,t_c)^{\ell_v}$ firings where $\ell_v$
	is the number of abstract transition firings
	occurring in $\sigma$ triggered by $v$. It is routine to check that
	$\widetilde{\sigma}$ is coverability sequence. The proof for
	$\mathcal L(\N',s[r',m'_0],S'_f)$ is similar.
		
	\medskip\noindent
	Let us prove that $\mathcal L(\widetilde{\N},s[r,p_0],\widetilde{S}_f)
	\subseteq \mathcal L(\N,s[\tilde{r},m_0],S_f) \cup \mathcal L(\N',s[r',m'_0],S'_f)$.
	Observe that any firing sequence must start
	by a firing of $t_b$ or $t'_b$.
	Let $t_b\widetilde{\sigma}$ be a coverability sequence of
	$(\widetilde{\N},s[\tilde{r},p_0],\widetilde{S}_f)$. Consider the sequence $\sigma$
	obtained by deleting all the firings of $t_c$ in $\tilde{\sigma}$.
	It is routine to check that
	$\sigma$ is a coverability sequence for $(\N,s[r,m_0],S_f)$.
	The case of a coverability sequence starting by $t'_b$ is similar.	
\end{proof}
The next theorem has two interesting consequences: the family of RPN coverability languages
is not closed under intersection with the family of regular languages. But the family obtained by this intersection
is \emph{quite close} to the family of recursively enumerable languages.
The result was already stated in Proposition~9 of~\cite{HaddadP99}
for the family of RPN reachability languages but the proof was only sketched.
\begin{theorem}
	\label{theo:re}
	Let $\mathcal L$ be a recursively enumerable language. Then there
	exist an RPN language $\mathcal L'$, a regular language
	$\mathcal R$ and a homomorphism $h$ such that
	$\mathcal L=h(\mathcal L'\cap\mathcal R)$.
\end{theorem}
\eject
\begin{proof}
	Let $\mathcal M=(\Sigma,L,\delta)$ be a Turing machine with it set of states $L$
	including $\ell_0$ (resp. $\ell_f$) the initial (resp. final) state and its transition function
	$\delta$ from $L\times \Sigma \cup \{\flat\}$ to $L\times \Sigma \times \{\leftarrow,\rightarrow\}$
	where $\flat$ is the blank character.
	
	Let us define a labeled marked RPN $\N$ and an automaton $\mathcal A$.
	Their common alphabet is the set of transitions of $\N$
	and the labeling of the transitions of the RPN is the identity mapping.
	The intersection of their languages is thus the language of the synchronized product of the two devices.
	The single final state of $\N$ (to be covered) is the empty tree.
	
	The automaton $\mathcal A$ is depicted below (with $\Sigma=\{a,b\}$). In $q_0$ it allows $\N$
	to generate the representation of any word $w\in \Sigma^*$, input of $\mathcal M$.
	However, this intermediate representation is not suitable for mimicking $\mathcal M$.
	Thus in $q_1$, the intermediate representation is translated into an appropriate one.
	Once this representation is obtained, it mimics any transition of $\mathcal M$
	by triggering the firing of several transitions of $\N$. We will detail
	this simulation after the specification of $\N$.
	
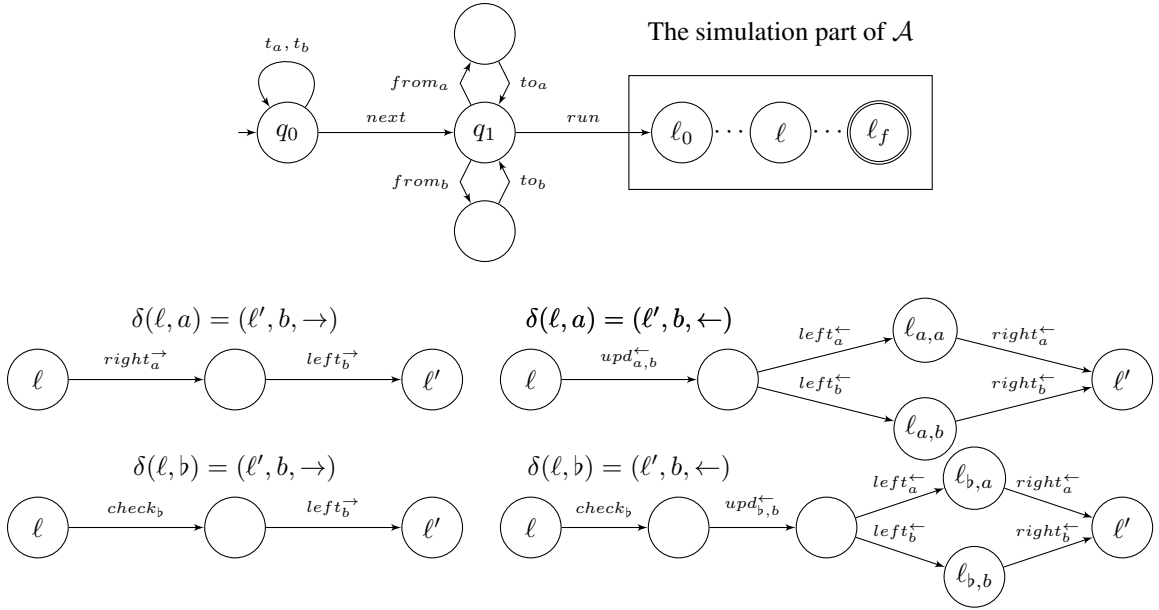
	\begin{center}
		\begin{tikzpicture}[xscale=0.65,yscale=0.65]
			
			\path (10,0) node[draw,rectangle,minimum height=1.5cm,minimum width=4cm] () {};
			
			\path (0,0) node[draw,circle,inner sep=2pt,minimum size=0.8cm] (q0) {\small{$q_0$}};
			\path (4,0) node[draw,circle,inner sep=2pt,minimum size=0.8cm] (q1) {\small{$q_1$}};
			\path (4,2) node[draw,circle,inner sep=2pt,minimum size=0.8cm] (q2) {};
			\path (4,-2) node[draw,circle,inner sep=2pt,minimum size=0.8cm] (q3) {};
			\path (8,0) node[draw,circle,inner sep=2pt,minimum size=0.8cm] (q4) {\small{$\ell_0$}};
			\path (10,0) node[draw,circle,inner sep=2pt,minimum size=0.8cm] (q5) {\small{$\ell$}};
			\path (12,0) node[draw,circle,accepting,inner sep=2pt,minimum size=0.8cm] (q6) {\small{$\ell_f$}};

			\path (10,2) node[]  {\small{The simulation part of $\mathcal A$}};
			
			\path (9,0) node[]  {\small{$\cdots$}};
			\path (11,0) node[]  {\small{$\cdots$}};
			
			\draw[arrows=-latex'] (-1,0) -- (q0)  ;
			
			\draw[arrows=-latex'] (q0) -- (q1) node[pos=0.5,above] {\tiny{$next$}} ;
			\draw[arrows=-latex'] (q1) -- (3.5,1)--(q2) node[pos=0,left] {\tiny{$from_a$}} ;
			\draw[arrows=-latex'] (q2) -- (4.5,1)--(q1) node[pos=0,right] {\tiny{$to_a$}} ;
			\draw[arrows=-latex'] (q1) -- (3.5,-1)--(q3) node[pos=0,left] {\tiny{$from_b$}} ;
			\draw[arrows=-latex'] (q3) -- (4.5,-1)--(q1) node[pos=0,right] {\tiny{$to_b$}} ;
			\draw[arrows=-latex'] (q1) -- (q4) node[pos=0.5,above] {\tiny{$run$}} ;

			\draw[-latex'] (q0) .. controls +(55:60pt) and +(125:60pt) .. (q0)
			node[pos=.5,above] {\tiny{$t_a,t_b$}};

		\end{tikzpicture}
	\end{center}
	
	\begin{center}
		\begin{tikzpicture}[xscale=0.65,yscale=0.65]
			
			\path (4,1.2) node[]  {\small{$\delta(\ell,a)=(\ell',b,\rightarrow)$}};

			\path (0,0) node[draw,circle,inner sep=2pt,minimum size=0.8cm] (q0) {\small{$\ell$}};
			\path (4,0) node[draw,circle,inner sep=2pt,minimum size=0.8cm] (q1) {};
			\path (8,0) node[draw,circle,inner sep=2pt,minimum size=0.8cm] (q2) {\small{$\ell'$}};
			
			\draw[arrows=-latex'] (q0) -- (q1) node[pos=0.5,above] {\tiny{$right_{a}^{\rightarrow}$}} ;
			\draw[arrows=-latex'] (q1) -- (q2) node[pos=0.5,above] {\tiny{$left_{b}^{\rightarrow}$}} ;

			\path (12,1.2) node[]  {\small{$\delta(\ell,a)=(\ell',b,\leftarrow)$}};

			\path (10,0) node[draw,circle,inner sep=2pt,minimum size=0.8cm] (q10) {\small{$\ell$}};
			\path (14,0) node[draw,circle,inner sep=2pt,minimum size=0.8cm] (q11) {};
			\path (18,1) node[draw,circle,inner sep=2pt,minimum size=0.8cm] (q12) {\small{$\ell_{a,a}$}};
			\path (18,-1) node[draw,circle,inner sep=2pt,minimum size=0.8cm] (r12) {\small{$\ell_{a,b}$}};
			\path (22,0) node[draw,circle,inner sep=2pt,minimum size=0.8cm] (q13) {\small{$\ell'$}};
			
			\draw[arrows=-latex'] (q10) -- (q11) node[pos=0.5,above] {\tiny{$upd_{a,b}^{\leftarrow}$}} ;
			\draw[arrows=-latex'] (q11) -- (q12) node[pos=0.5,above] {\tiny{$left_{a}^{\leftarrow}$}} ;
			\draw[arrows=-latex'] (q12) -- (q13) node[pos=0.5,above] {\tiny{$right_{a}^{\leftarrow}$}} ;
			\draw[arrows=-latex'] (q11) -- (r12) node[pos=0.5,above] {\tiny{$left_{b}^{\leftarrow}$}} ;
			\draw[arrows=-latex'] (r12) -- (q13) node[pos=0.5,above] {\tiny{$right_{b}^{\leftarrow}$}} ;
			
			\path (4,-1.8) node[]  {\small{$\delta(\ell,\flat)=(\ell',b,\rightarrow)$}};

			\path (0,-3) node[draw,circle,inner sep=2pt,minimum size=0.8cm] (q20) {\small{$\ell$}};
			\path (4,-3) node[draw,circle,inner sep=2pt,minimum size=0.8cm] (q21) {};
			\path (8,-3) node[draw,circle,inner sep=2pt,minimum size=0.8cm] (q22) {\small{$\ell'$}};
			
			\draw[arrows=-latex'] (q20) -- (q21) node[pos=0.5,above] {\tiny{$check_{\flat}$}} ;
			\draw[arrows=-latex'] (q21) -- (q22) node[pos=0.5,above] {\tiny{$left_{b}^{\rightarrow}$}} ;
			
			\path (12,1.2) node[]  {\small{$\delta(\ell,a)=(\ell',b,\leftarrow)$}};

			\path (10,-3) node[draw,circle,inner sep=2pt,minimum size=0.8cm] (q30) {\small{$\ell$}};
			\path (13,-3) node[draw,circle,inner sep=2pt,minimum size=0.8cm] (q30b) {};
			\path (16,-3) node[draw,circle,inner sep=2pt,minimum size=0.8cm] (q31) {};
			\path (19,-2) node[draw,circle,inner sep=2pt,minimum size=0.8cm] (q32) {\small{$\ell_{\flat,a}$}};
			\path (19,-4) node[draw,circle,inner sep=2pt,minimum size=0.8cm] (r32) {\small{$\ell_{\flat,b}$}};
			\path (22,-3) node[draw,circle,inner sep=2pt,minimum size=0.8cm] (q33) {\small{$\ell'$}};
			
			\draw[arrows=-latex'] (q30) -- (q30b) node[pos=0.5,above] {\tiny{$check_{\flat}$}} ;
			\draw[arrows=-latex'] (q30b) -- (q31) node[pos=0.5,above] {\tiny{$upd_{\flat,b}^{\leftarrow}$}} ;
			\draw[arrows=-latex'] (q31) -- (q32) node[pos=0.5,above] {\tiny{$left_{a}^{\leftarrow}$}} ;
			\draw[arrows=-latex'] (q32) -- (q33) node[pos=0.5,above] {\tiny{$right_{a}^{\leftarrow}$}} ;
			\draw[arrows=-latex'] (q31) -- (r32) node[pos=0.5,above] {\tiny{$left_{b}^{\leftarrow}$}} ;
			\draw[arrows=-latex'] (r32) -- (q33) node[pos=0.5,above] {\tiny{$right_{b}^{\leftarrow}$}} ;
			
			\path (12,-1.8) node[]  {\small{$\delta(\ell,\flat)=(\ell',b,\leftarrow)$}};
			
		\end{tikzpicture}
	\end{center}

	$\mathcal N$ is defined as follows.
	Its set of places is $P=\{p_a \mid a \in \Sigma\}\cup \{root,right,left,start,ret\}$.
	We now define the set of transitions $T$. The first subset
	corresponds to the generation of a representation of the input word of  $\mathcal M$.
	\begin{itemize}[nosep]\smallskip
		\item For all $a\in \Sigma$, $t_a\in T_{ab}$ with $W^-(t_a)=start$, $W^+(t_a)=ret$ and $\Omega(t_a)=start+p_a$;
		\item $next\in T_{el}$ with  $W^-(next)=start$ and $W^+(next)=ret$;
		\item For all  $a\in \Sigma$, $from_a\in T_{\tau}$  $W^-(from_a)=ret+p_a$;
		\item For all  $a\in \Sigma$, $to_a\in T_{ab}$ with  $W^-(to_a)=right$, $W^+(to_a)=right$ and $\Omega(to_a)=right+p_a$;
		\item $run\in T_{el}$ with $W^-(run)=root+ret$ and $W^+(run)=root$
	\end{itemize}\medskip
	
	The second subset corresponds to the simulation of $\mathcal M$.
	\begin{itemize}[nosep]\smallskip
		\item For all $a\in \Sigma$, $right_{a}^{\rightarrow}\in T_{\tau}$ with $W^-(right_{a}^{\rightarrow})=right+p_a$;
		\item For all $a\in \Sigma$,  $left_{a}^{\rightarrow}\in T_{ab}$ with $W^-(left_{a}^{\rightarrow})=W^+(left_{a}^{\rightarrow})=left$\\
		and $\Omega(left_{a}^{\rightarrow})=left+p_a$;
		\item For all $a,b \in \Sigma$,
		$upd_{a,b}^{\leftarrow}\in T_{el}$ with $W^-(upd_{a,b}^{\leftarrow})=right+p_a$ and $W^+(upd_{a,b}^{\leftarrow})=right+p_{b}$
		
		\item For all $a\in \Sigma$, $left_{a}^{\leftarrow}\in T_{\tau}$ with $W^-(left_{a}^{\leftarrow})=left+p_a$
		\item For all $a\in \Sigma$, $right_{a}^{\leftarrow}\in T_{ab}$ with $W^-(right_{a}^{\leftarrow})=W^+(right_{a}^{\leftarrow})=right$\\
		and $\Omega(right_{a}^{\leftarrow})=right+p_a$
		\item $check_{\flat}\in T_{el}$ with  $W^-(check_{\flat})=W^+(check_{\flat})=right+root$;
		\item For all $b \in \Sigma$,
		$upd_{\flat,b}^{\leftarrow}\in T_{ab}$ with $W^-(upd_{\flat,b}^{\leftarrow})=right$, $W^+(upd_{\flat,b}^{\leftarrow})=right$\\
		and $\Omega(upd_{\flat,b}^{\leftarrow})=right+p_b$.
	\end{itemize}
	The initial state is $s[r,root+start+left+right]$.
	
\medskip
	Let us explain how the simulation works.
	Let $abc$ be the word on the tape of $\mathcal M$. Then firing $(r,t_a)(v_1,t_b)(v_2,t_c)$ one gets:
	
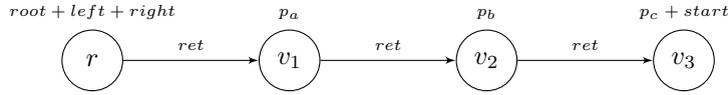
\begin{center}
		\begin{tikzpicture}[xscale=0.65,yscale=0.65]
			
			\path (0,0) node[draw,circle,inner sep=2pt,minimum size=0.8cm,
			label={[xshift=0cm, yshift=0cm]\tiny{$root+left+right$}}] (q0) {\small{$r$}};
			\path (4,0) node[draw,circle,inner sep=2pt,minimum size=0.8cm,
			label={[xshift=0cm, yshift=0cm]\tiny{$p_a$}}] (q1) {\small{$v_1$}};
			\path (8,0) node[draw,circle,inner sep=2pt,minimum size=0.8cm,
			label={[xshift=0cm, yshift=0cm]\tiny{$p_b$}}] (q2) {\small{$v_2$}};
			\path (12,0) node[draw,circle,inner sep=2pt,minimum size=0.8cm,
			label={[xshift=0cm, yshift=0cm]\tiny{$p_c+start$}}] (q3) {\small{$v_3$}};
			
			\draw[arrows=-latex'] (q0) -- (q1) node[pos=0.5,above] {\tiny{$ret$}} ;
			\draw[arrows=-latex'] (q1) -- (q2) node[pos=0.5,above] {\tiny{$ret$}} ;
			\draw[arrows=-latex'] (q2) -- (q3) node[pos=0.5,above] {\tiny{$ret$}} ;
			
		\end{tikzpicture}
	\end{center}
	
	After firing  $(v_3,next)(v_3,from_c)(r,to_c)(v_2,from_b)(u_1,to_b)(v_1,from_a)(u_2,to_a)(r,run)$ one gets:
	
	
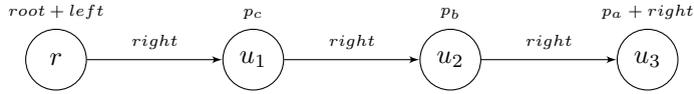
\begin{center}
		\begin{tikzpicture}[xscale=0.65,yscale=0.65]
			
			\path (0,0) node[draw,circle,inner sep=2pt,minimum size=0.8cm,
			label={[xshift=0cm, yshift=0cm]\tiny{$root+left$}}] (q0) {\small{$r$}};
			\path (4,0) node[draw,circle,inner sep=2pt,minimum size=0.8cm,
			label={[xshift=0cm, yshift=0cm]\tiny{$p_c$}}] (q1) {\small{$u_1$}};
			\path (8,0) node[draw,circle,inner sep=2pt,minimum size=0.8cm,
			label={[xshift=0cm, yshift=0cm]\tiny{$p_b$}}] (q2) {\small{$u_2$}};
			\path (12,0) node[draw,circle,inner sep=2pt,minimum size=0.8cm,
			label={[xshift=0cm, yshift=0cm]\tiny{$p_a+right$}}] (q3) {\small{$u_3$}};
			
			\draw[arrows=-latex'] (q0) -- (q1) node[pos=0.5,above] {\tiny{$right$}} ;
			\draw[arrows=-latex'] (q1) -- (q2) node[pos=0.5,above] {\tiny{$right$}} ;
			\draw[arrows=-latex'] (q2) -- (q3) node[pos=0.5,above] {\tiny{$right$}} ;
			
		\end{tikzpicture}
	\end{center}
	
	Let us describe the two cases of tape simulation. Assume that the content of
	the tape is $abcd\flat^\omega$ and that the head of $\mathcal M$ is over $c$ then
	the corresponding state is the following one. The ``left'' branch contains the content of the tape
	on the left of the head while descending to the leaf and the ``right'' branch contains the relevant content of the tape
	on the right of the head (including the cell under the head) while ascending from the leaf. Thus
	the token in place $right$ points to the thread corresponding to the cell under the head
	while the token in place $left$ points to the thread corresponding to the cell immediately on the left
	of the head. The state of $\mathcal M$ is the state of $\mathcal A$.
	
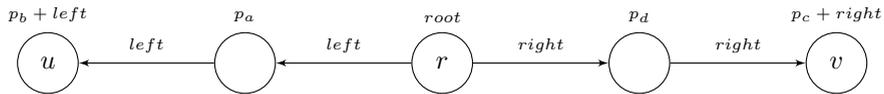
	\begin{center}
		\begin{tikzpicture}[xscale=0.65,yscale=0.65]
			
			\path (0,0) node[draw,circle,inner sep=2pt,minimum size=0.8cm,
			label={[xshift=0cm, yshift=0cm]\tiny{$root$}}] (q0) {\small{$r$}};
			\path (4,0) node[draw,circle,inner sep=2pt,minimum size=0.8cm,
			label={[xshift=0cm, yshift=0cm]\tiny{$p_d$}}] (q1) {};
			\path (8,0) node[draw,circle,inner sep=2pt,minimum size=0.8cm,
			label={[xshift=0cm, yshift=0cm]\tiny{$p_c+right$}}] (q2) {\small{$v$}};
			\path (-4,0) node[draw,circle,inner sep=2pt,minimum size=0.8cm,
			label={[xshift=0cm, yshift=0cm]\tiny{$p_a$}}] (q3) {};
			\path (-8,0) node[draw,circle,inner sep=2pt,minimum size=0.8cm,
			label={[xshift=0cm, yshift=0cm]\tiny{$p_b+left$}}] (q4) {\small{$u$}};
			
			\draw[arrows=-latex'] (q0) -- (q1) node[pos=0.5,above] {\tiny{$right$}} ;
			\draw[arrows=-latex'] (q1) -- (q2) node[pos=0.5,above] {\tiny{$right$}} ;
			\draw[arrows=-latex'] (q0) -- (q3) node[pos=0.5,above] {\tiny{$left$}} ;
			\draw[arrows=-latex'] (q3) -- (q4) node[pos=0.5,above] {\tiny{$left$}} ;
			
		\end{tikzpicture}
	\end{center}
	
	Assume that the content of
	the tape is $abcd\flat^\omega$ and that the head of $\mathcal M$ is over the first $\flat$ then
	the corresponding state is the following one.
	
	
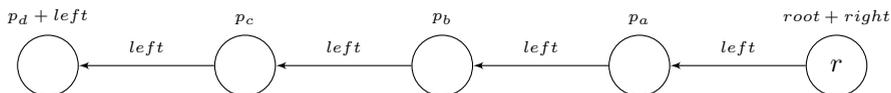
\begin{center}
		\begin{tikzpicture}[xscale=0.65,yscale=0.65]
			
			\path (0,0) node[draw,circle,inner sep=2pt,minimum size=0.8cm,
			label={[xshift=0cm, yshift=0cm]\tiny{$root+right$}}] (q0) {\small{$r$}};
			\path (-4,0) node[draw,circle,inner sep=2pt,minimum size=0.8cm,
			label={[xshift=0cm, yshift=0cm]\tiny{$p_a$}}] (q1) {};
			\path (-8,0) node[draw,circle,inner sep=2pt,minimum size=0.8cm,
			label={[xshift=0cm, yshift=0cm]\tiny{$p_b$}}] (q2) {};
			\path (-12,0) node[draw,circle,inner sep=2pt,minimum size=0.8cm,
			label={[xshift=0cm, yshift=0cm]\tiny{$p_c$}}] (q3) {};
			\path (-16,0) node[draw,circle,inner sep=2pt,minimum size=0.8cm,
			label={[xshift=0cm, yshift=0cm]\tiny{$p_d+left$}}] (q4) {};
			
			\draw[arrows=-latex'] (q0) -- (q1) node[pos=0.5,above] {\tiny{$left$}} ;
			\draw[arrows=-latex'] (q1) -- (q2) node[pos=0.5,above] {\tiny{$left$}} ;
			\draw[arrows=-latex'] (q2) -- (q3) node[pos=0.5,above] {\tiny{$left$}} ;
			\draw[arrows=-latex'] (q3) -- (q4) node[pos=0.5,above] {\tiny{$left$}} ;
			
		\end{tikzpicture}
	\end{center}
	It is routine to check that the simulation works. Let us illustrate it with one example.
	Assume that the content of
	the tape is $abcd\flat^\omega$,  the head of $\mathcal M$ is over $c$ and the current
	state is $\ell$. Let $\delta(\ell,c)=(\ell',e,\leftarrow)$.
	Then after firing $(v,upd_{c,e}^{\leftarrow})(u,left_{b}^{\leftarrow})(v,right_{b}^{\leftarrow})$, one gets:
	
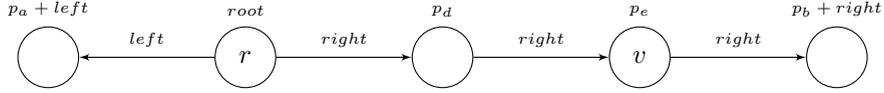
\begin{center}
		\begin{tikzpicture}[xscale=0.65,yscale=0.65]
			
			\path (0,0) node[draw,circle,inner sep=2pt,minimum size=0.8cm,
			label={[xshift=0cm, yshift=0cm]\tiny{$root$}}] (q0) {\small{$r$}};
			\path (4,0) node[draw,circle,inner sep=2pt,minimum size=0.8cm,
			label={[xshift=0cm, yshift=0cm]\tiny{$p_d$}}] (q1) {};
			\path (8,0) node[draw,circle,inner sep=2pt,minimum size=0.8cm,
			label={[xshift=0cm, yshift=0cm]\tiny{$p_e$}}] (q2) {\small{$v$}};
			\path (-4,0) node[draw,circle,inner sep=2pt,minimum size=0.8cm,
			label={[xshift=0cm, yshift=0cm]\tiny{$p_a+left$}}] (q3) {};
			\path (12,0) node[draw,circle,inner sep=2pt,minimum size=0.8cm,
			label={[xshift=0cm, yshift=0cm]\tiny{$p_b+right$}}] (q4) {};
			
			\draw[arrows=-latex'] (q0) -- (q1) node[pos=0.5,above] {\tiny{$right$}} ;
			\draw[arrows=-latex'] (q1) -- (q2) node[pos=0.5,above] {\tiny{$right$}} ;
			\draw[arrows=-latex'] (q0) -- (q3) node[pos=0.5,above] {\tiny{$left$}} ;
			\draw[arrows=-latex'] (q2) -- (q4) node[pos=0.5,above] {\tiny{$right$}} ;
			
		\end{tikzpicture}
	\end{center}
	For all $a\in \Sigma$, the homomorphism $h$ maps $t_a$ to $a$
	and for all $t \notin \{t_a\}_{a\in \Sigma}$, $h$ maps $t$ to $\varepsilon$.
	\end{proof}

Obviously, the family of RPNs coverability languages include the family of PNs coverability languages.
In~\cite{HP-icatpn99}, Proposition~1 establishes that the family of context-free languages is included
in family of reachability languages for RPNs. The proof relies on simulating the leftmost
derivations of a context-free grammar within particular two places $b_X$ and $e_X$
per nonterminal symbol $X$ where a token in $b_X$ means that $X$ must derived and
a token in $e_X$ means that the derivation of $X$ into a word has been achieved. In order to adapt
this result for the family of coverability languages for RPNs, it is enough to consider w.l.o.g.
that the initial symbol $I$ never appears on the right hand side of a rule and to specify
$s[r,e_I]$ as final state. We refer the reader to~\cite{HP-icatpn99} for more details.

\begin{restatable}{proposition}{CFinRPN}
	\label{prop:CFinRPN}
	The family of Context-free languages is included
	in the family of coverability languages of RPNs.
\end{restatable}
Since  universality is undecidable for the family of context-free languages, we deduce that universality of the family of RPN coverability languages is undecidable.	

\medskip
Let  $\mathcal L_1=\{a^mb^nc^p \mid m\geq n \geq p\}$.
Denote by $\mathcal L_2 = \{w\tilde{w}\mid w\in \{d,e\}^*\}$ where $\tilde{w}$ is the mirror of $w$.
Let  $\mathcal L_3=\{a^nb^nc^n \mid  n \in \nat\}$.
Observe that given the final marking $p_f$ we get that the net in Figure~\ref{fig:PN_for_L1_and_L3} has $\mathcal L_1$ as its coverability language, and $\mathcal L_3$ its reachability language.

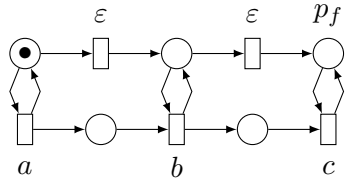
\begin{figure}[!h]
		\begin{center}
		\begin{tikzpicture}[xscale=1,yscale=1]
			
			\path (0,0) node[] {$\bullet$};
			\path (0,0) node[draw,circle,inner sep=2pt,minimum size=0.4cm] (p1) {};
			\path (1,0.5) node[] {$\varepsilon$};
			\path (1,0) node[draw,rectangle,inner sep=2pt,minimum width=0.2cm,minimum height=0.4cm] (t1) {};
			\path (2,0) node[draw,circle,inner sep=2pt,minimum size=0.4cm] (p2) {};
			\path (3,0.5) node[] {$\varepsilon$};
			\path (3,0) node[draw,rectangle,inner sep=2pt,minimum width=0.2cm,minimum height=0.4cm] (t2) {};
			\path (4,0) node[draw,circle,inner sep=2pt,minimum size=0.4cm] (p3) {};
			\path (4,0.5) node[] {$p_f$};
			
			\path (0,-1.5) node[] {$a$};
			\path (0,-1) node[draw,rectangle,inner sep=2pt,minimum width=0.2cm,minimum height=0.4cm] (t3) {};
			\path (1,-1) node[draw,circle,inner sep=2pt,minimum size=0.4cm] (p4) {};
			
			\path (2,-1.5) node[] {$b$};
			\path (2,-1) node[draw,rectangle,inner sep=2pt,minimum width=0.2cm,minimum height=0.4cm] (t4) {};
			\path (3,-1) node[draw,circle,inner sep=2pt,minimum size=0.4cm] (p5) {};
			
			\path (4,-1.5) node[] {$c$};
			\path (4,-1) node[draw,rectangle,inner sep=2pt,minimum width=0.2cm,minimum height=0.4cm] (t5) {};
			
			\draw[arrows=-latex] (p1) -- (t1) ;
			\draw[arrows=-latex] (t1) -- (p2) ;
			\draw[arrows=-latex] (p2) -- (t2) ;
			\draw[arrows=-latex] (t2) -- (p3) ;
			
			\draw[arrows=-latex] (t3) -- (p4) ;
			\draw[arrows=-latex] (p4) -- (t4) ;
			
			\draw[arrows=-latex] (t4) -- (p5) ;
			\draw[arrows=-latex] (p5) -- (t5) ;
			
			\draw[arrows=-latex] (p1) -- (-0.2,-0.5)-- (t3) ;
			\draw[arrows=-latex] (t3) -- (0.2,-0.5)-- (p1) ;
			
			\draw[arrows=-latex] (p2) -- (1.8,-0.5)-- (t4) ;
			\draw[arrows=-latex] (t4) -- (2.2,-0.5)-- (p2) ;
			
			\draw[arrows=-latex] (p3) -- (3.8,-0.5)-- (t5) ;
			\draw[arrows=-latex] (t5) -- (4.2,-0.5)-- (p3) ;		
		\end{tikzpicture}
	\end{center}\vspace*{-7mm}
\caption{A Petri net for the languages $\mathcal L_1 $ and $ \mathcal L_3 $}
\label{fig:PN_for_L1_and_L3}
\end{figure}

The next proposition witnesses a Petri net language interesting from an expressiveness
point of view. A similar result can be found page~179 in Peterson's book~\cite{nla.cat-vn2956435}.
\begin{restatable}{proposition}{PNnotCF}
	\label{prop:PNnotCF}
	$\mathcal L_1$ is the coverability language of some Petri net but it
	is not a context-free language.
\end{restatable}

\begin{proof}
	Let us recall (a weak version of) Ogden lemma~\cite{ogden1968helpful}. For any context-free language
	$\mathcal L\,$there \mbox{exists} an integer $N$ such for any word $w\! \in\! \mathcal L$
	with $N$ marked positions, there exists a decomposition
	$w\!=\!w_1w_2w_3w_4w_5$ such that $w_2w_4$ contains at least a marked position and
	for all $n\!\geq\! 0$, $w_1w_2^nw_3w_4^nw_5\!\in\!\mathcal L$.
	
	\noindent
	The proof that $\mathcal L_1$ is not a context-free language is similar to the proof of the folk
	result that $\mathcal L_3$ is not a context-free language.
	Assume that $\mathcal L_1$ is a context-free language and consider
	the word $w=a^Nb^Nc^N$ with all $c$ positions marked.
	So let $w=w_1w_2w_3w_4w_5$ with the decomposition fulfilling the requirements
	of Ogden lemma. Since $w'=w_1w_2^2w_3w_4^2w_5\in \mathcal L_1$,
	$w_2$ and $w_4$ are mono-letter words. Furthermore one of these words
	is equal to $c^q$ for some $q>0$. If $w_2=c^q$ then $w_4=c^{q'}$
	and thus $w'$ contains too much $c$'s to belong to $\mathcal L_1$. If $w_4=c^q$ then either
	$w_2=a^{q'}$, $w_2=b^{q'}$ or $w_2=c^{q'}$. Whatever the case, $w'$ misses either $a$'s or $b$'s
	to belong to $\mathcal L_1$.
	\noindent
	As mentioned before the coverability language for the net in Figure~\ref{fig:PN_for_L1_and_L3} with final marking $p_f$ is $\mathcal L_1$.
\end{proof}

Using the previous results, the next theorem emphasises
the expressive power of coverability languages of RPNs.
\begin{theorem}\label{context-free}
	The family of coverability languages of RPNs strictly include
	the union of the family of coverability languages of PNs
	and the family of context-free languages.
\end{theorem}
\begin{proof}
	The inclusion is an immediate consequence of Proposition~\ref{prop:CFinRPN}.
	Consider the language
	$\mathcal L=\mathcal L_1 \cup \mathcal L_2$.
	
	\noindent
	Since (1) by Proposition~\ref{prop:closedunion}, the family of
	coverability languages of RPNs is closed under union,
	(2) $\mathcal L_1$  is a PN language, and
	(3) the language
	of palindromes is a context-free language, we deduce that $\mathcal L$ is an RPN language.
	
	\noindent
	PN and context-free languages are closed under homomorphism.
	Since the projection of $\mathcal L$ on $\{a,b,c\}$ is the language
	of Proposition~\ref{prop:PNnotCF}, $\mathcal L$ is not a context-free language.
	The projection of $\mathcal L$ on $\{d,e\}$ is the language
	of palindromes. Since it was seen in~\cite{Lambert92}  that the language of (2 letters) palindromes
	is not a coverability language for any PN we are done.
\end{proof}

The next propositions show that the family of coverability languages of an RPN
is a particular family of reachability languages of an RPN : the family of \emph{cut languages}.
A cut language of an RPN is a reachability language with a single final state $\emptyset$.

\begin{proposition}
\label{prop:reach-inc-cov}
The family of cut languages of RPNs
is included in the family of coverability languages of RPNs.
\end{proposition}
\begin{proof}
        Due to the correspondence between firing sequences of $(\N,s_0)$ and those of $(\rooted{\N},\sr)$,
        established in the previous section,
        one can assume w.l.o.g. that the initial markings of the RPNs have a single vertex.
Let $\mathcal L_R(\N,s[r,m_0],\{\emptyset\})$ be such a reachability language.\\
$\N'$ is obtained by adding places $todo$ and $done$
and a transition $start \in T'_{ab}$ with:\\
\centerline{$\lambda'(start)=\varepsilon$, $W'^-(start)=todo$, $W'^+(start)=done$,  $\Omega'(start)=m_0$.}\\
Then it is routine
to check that $\mathcal L_C(\N',s[r,todo],\{s[r,done]\})=\mathcal L_R(\N,s[r,m_0],\{\emptyset\})$.
\end{proof}

Establishing the converse inclusion is more intricate.

\begin{proposition}
\label{prop:cov-inc-reach}
The family of coverability languages of RPNs is included
in the family of cut languages of RPNs.
\end{proposition}
\begin{proof}
        Due to the correspondence between firing sequences of $(\N,s_0)$ and those of $(\rooted{\N},\sr)$,
        established in the previous section,
        one can assume w.l.o.g. that the initial markings of the RPNs have a single vertex.
Let $\mathcal L_C(\N,s[r,m_0],S_f)$ be a coverability RPN language.

\medskip\noindent {\bf Case $\emptyset \in S_f$.}
Observe that in this case we can reduce  $S_f$ to $\{\emptyset\}$.
Then  $\N'$ is obtained from $\N$ by adding a place $root$
and a cut transition $t_{root}$ with $\lambda'(t_{root})=\varepsilon$ and $W'^{-}(t_{root})=root$. It is routine
to check that the reachability language
$\mathcal L_R(\N',s[r,m_0+root],\{\emptyset\})=\mathcal L_C(\N,s[r,m_0],\{\emptyset\})$.

\medskip
\noindent {\bf Case $\emptyset \notin S_f$.}
Consider the net $\N^*$ obtained from $\N$ by adding two places $start$ and $run$
with  $m^*_0=start$, transitions $t_{run}\in T_{el}$
and $t_{start}\in T_{ab}$ with $\lambda^*(t_{run})=\lambda^*(t_{start})=\varepsilon$ and:\\
$W^{*-}(t_{run})=run$, $W^{*+}(t_{run})=2run$,\\
$W^{*-}(t_{start})=start$, $W^{*+}(t_{start})={\bf 0}$ and $\Omega^*(t_{start})=m_0+run$.
\begin{itemize}[nosep]
  \item For all $t\in T_{el}$,
           $W^{*-}(t)=W^{*-}(t)+run$ and $W^{*+}(t)=W^{+}(t)$;
  \item For all $t\in T_{ab}$,  $\Omega^*(t)=\Omega(t)+run$,
           $W^{*-}(t)=W^{-}(t)+run$ and $W^{*+}(t)=W^{+}(t)+run$;
  \item For all $t\in T_{\tau}$,  $W^{*-}(t)=W^{-}(t)+run$.
\end{itemize}
Let $S^*_f$ be $S_f$ where all markings are increased by $run$.\\
Then it is routine
to check that:
$\mathcal L_C(\N^*,s[r,m^*_0],S^*_f)=\mathcal L_C(\N,s[r,m_0],S_f)$.\\
Furthermore (1) the empty tree is not reachable in $(\N^*,s[r,m^*_0])$
and (2) for any coverability sequence $s[r,m^*_0] \xrightarrow{\sigma} s\succeq s_f \in S^*_f$,
$r$ does not belong to the image of the corresponding mapping $f$.
Thus in the rest of the proof
we assume that $(\N,s[r,m_0],S_f)$ fulfills these properties.
We also assume
w.l.o.g. that all vertices in $S_f$ are distinct. We denote $V_f$ this set of vertices.

\medskip
\noindent
Let $\N'$ obtained as follows.\\
One adds  places $todo,done,cut,\{p_v\mid v\in V_f\}, \{p_{u,v}\mid s \in S_f,u\xrightarrow{m_v}_s v\}$.
\begin{itemize}[nosep]
  \item For all $t\in T_{el}$,
           $W'^{-}(t)=W^{-}(t)$ and $W'^{+}(t)=W^{+}(t)$;
  \item For all $t\in T_{ab}$,
           $W'^{-}(t)=W^{-}(t)$, $W'^{+}(t)=W^{+}(t)$ and $\Omega'(t)=\Omega(t)+cut$;
  \item For all $t\in T_{\tau}$,  $W'^{-}(t)=W^{-}(t)+cut$.
\end{itemize}
For all $t \in T_{ab}$, one adds the following abstract transitions:
\begin{itemize}[nosep]
	\item one adds $t_{Br}\in T'_{ab}$ with $\lambda'(t_{Br})=\lambda(t)$ and\\
	$W'^-(t_{Br})=W^-(t)+todo$, $W'^+(t_{Br})=done$, $\Omega'(t_{Br})=\Omega(t)+todo$;
	\item For all $r_s$ with $s\in S_f$
	one adds  $t_{r_s}\in T'_{ab}$ with $\lambda'(t_{r_s})=\lambda(t)$ and\\
	$W'^-(t_{r_s})=W^-(t)+todo$, $W'^+(t_{r_s})=done$, $\Omega'(t_{r_s})=\Omega(t)+(|\{r_s\xrightarrow{m_w}_s w\}|+1)p_{r_s}$;
	\item For all $v \in V_s\setminus\{r_s\}$ with $s\in S_f$ and $u \xrightarrow{m_v}_s v$ such that $W^+(t)\geq m_v$,\\
	one adds  $t_{v}\in T'_{ab}$ with $\lambda'(t_{v})=\lambda(t)$ and\\
	$W'^-(t_{v})=W^-(t)+p_u$, $W'^+(t_{v})=p_{u,v}$, $\Omega(t_{v})=\Omega(t)+(|\{v\xrightarrow{m_w}_s w\}|+1)p_{v}$.
\end{itemize}
One adds the following cut transitions:
\begin{itemize}[nosep]
	\item One adds $\tau_{done}\in  T_{\tau}$ with $W'^{-}(\tau_{done})=done$ and $\lambda'(\tau_{done})=\varepsilon$.
	\item For all $v \in V_s$ with $s\in S_f$
	one adds  $\tau_{v}\in T'_{\tau}$ with $\lambda'(\tau_{v})=\varepsilon$ and\\
	$W'^-(\tau_{v})=M_s(v)+p_{v}+\sum_{v\xrightarrow{m_w}_s w}p_{v,w}$.
\end{itemize}
\noindent
Let us prove that
$\mathcal L_R(\N',s[r,m_0+todo],\{\emptyset\})=\mathcal L_C(\N,s[r,m_0],S_f)$.

\medskip\noindent
$\bullet$ $\mathcal L_C(\N,s[r,m_0],S_f)\subseteq \mathcal  L_R(\N',s[r,m_0+todo],\{\emptyset\})$.
Consider in $\N$ a coverability sequence $s[r,m_0] \xrightarrow{\sigma} s\succeq s_f \in S_f$ with
$f$ the mapping from $V_{s_f}$ to $V_s$. Let $Br$ be the branch in $s$ from $r$ to $f(r_{s_f})$,
excluding $f(r_{s_f})$.
We build a sequence $\sigma'$ as follows.
\begin{itemize}[nosep]
	\item Let $v \in Br\setminus\{r\}$ and $(u,t)$ be the firing in $\sigma$ that creates $v$.\\
	 Then we substitute $(u,t)$ by $(u,t_{Br})$;
	\item Let $(u,t)$ be the firing in $\sigma$ that creates $f(r_{s_f})$.
	 Then we substitute $(u,t)$ by $(u,t_{r_{s_f}})$;
	\item Let $v \in V_{s_f}\setminus \{r_{s_f}\}$ and  $(u,t)$ be the firing in $\sigma$ that creates $f(v)$.\\
	Then we substitute $(u,t)$ by $(u,t_{r_{s_f}})$.
\end{itemize}\smallskip\noindent
Then $\sigma'$ is a firing sequence of $(\N',s[r,m_0+todo])$ that leads to $s'$ with the
same tree structure (and vertices) as the one of $s$ and where the markings labelling $s'$
are defined as follows.
\begin{itemize}[nosep]
	\item For all $v \in V_{s'}\setminus (Br\cup f(V_{s_f}))$, $M_{s'}(v)=M_{s}(v)+cut$,\\
	and all $u\xrightarrow{m'_v}_{s'} v$ and  $u\xrightarrow{m_v}_{s} v$, one has $m'_v=m_v$;
	\item  For all $v \in Br$, $M_{s'}(v)=M_{s}(v)$. For all $v\xrightarrow{m'_w}_{s'} w$ with $w \in Br \cup \{f(r_{s_f})\}$,
	$m'_w=done$;
	\item For all $v \in V_{s_f}$, $M_{s'}(f(v))=M_{s}(f(v))+p_v$.
	For all $f(v)\xrightarrow{m'_w}_{s'} f(w)$,
	$m'_w=p_{v,w}$.
\end{itemize}\medskip\noindent
Observe that $\lambda(\sigma')=\lambda(\sigma)$. Then one completes $\sigma'$ by
firing $\{(f(v),\tau_v)\}_{v\in V_{s_f}}$ bottom up followed by firing $\{(v,\tau_{done})\}_{v\in Br}$ bottom up
leading to $\emptyset$.

\medskip\noindent
$\bullet$ $\mathcal L_R(\N',s[r,m_0+todo],\{\emptyset\})\subseteq \mathcal L_C(\N,s[r,m_0],S_f)$.
Observe that in $(\N',s[r,m_0+todo])$ the only way to reach $\emptyset$ is to fire $\tau_{done}$ since in $r$
(by induction) only abstract transitions of $T_{ab}$, $\{t_{Br}\mid t\in T_{ab}\}$ and
$\{t_{r_s}\mid t\in T_{ab} \wedge s\in S_f\}$ are fireable
and places $cut$ and $\{p_v\}_{v\in V_f}$ are initially unmarked. Furthermore a single firing
$\{t_{Br}\mid t\in T_{ab}\}$ and
$\{t_{r_s}\mid t\in T_{ab} \wedge s\in S_f\}$ is at most possible in $r$ since
no transition can produce tokens for $todo$ in $r$.

\medskip\noindent
So consider in $\N'$ a firing sequence $s[r,m_0+todo] \xrightarrow{\sigma'} \emptyset$.
Due to the previous observation before the firing $(r,\tau_{done})$ ending $\sigma'$, there has been  in $\sigma'$ a firing
of $(r,t_{Br})$ or $(r,t_{r_s})$ for some $t\in T_{ab}$ and $s\in S_f$ creating a vertex $v_1$ followed by the firing of a cut transition in $v_1$.
Since $\Omega'(t_{Br})=\Omega(t)+todo$, if  $v_1$ has been created by $(r,t_{Br})$ then the only cut transition that can be fired in $v_1$
is $\tau_{done}$. Since $\lambda'(\tau_{done})=\varepsilon$ and $W'^+(t_{Br})=done$, this firing can delayed in $\sigma'$ just before
the firing of $(r,\tau_{done})$.

\noindent
Furthermore there
must have been before this firing, the firing of  $(v_1,t_{Br})$ or $(v_2,t_{r_s})$ for some $t\in T_{ab}$ and $s\in S_f$ creating a vertex $v_2$
followed by the firing of a cut transition in $v_2$.
Since this iterated reasoning must  end, there must be some $v_k$ created by the firing of $(v_{k-1},t_{r_s})$ (with $v_0=r$)
for some $t\in T_{ab}$ and  $s\in S_f$. We denote by $f(r_s)$ the vertex $v_k$.

\noindent
Since  $\Omega'(t_{r_s})=\Omega(t)+(|\{r_s\xrightarrow{m_w}_s w\}|+1)p_{r_s}$, the only  cut transition that can be fired in $f(r_s)$ is $\tau_{r_s}$. Since $\lambda'(\tau_{r_s})=\varepsilon$ and $W'^+(t_{r_s})=done$, this firing can delayed
in $\sigma'$ just before the firing of $(v_{k-1},\tau_{done})$.
Furthermore  the firing of this cut transition must have been preceded for all $r_s\xrightarrow{m_w}_s w$
by the firing of some abstract transition $(v_k,t_w)$ creating a vertex denoted $f(w)$ followed by the firing
of a cut transition in $f(w)$.

\noindent
Applying the same reasoning for $f(w)$ as the one for $f(r_s)$, one gets that
the only  cut transition that can be fired in $f(w)$ is $\tau_{w}$ and that all the firings related to these $w$'s
can be delayed before the firing  $(f(r_s),\tau_rs)$.

\noindent
Iterating this process, one obtains that $\sigma'$ can be reordered as $\sigma''\sigma_{\tau}$ with
$\lambda'(\sigma'')=\lambda'(\sigma')$, and $\sigma_{\tau}$ is a sequence of
cut transition firings with $\lambda(\sigma_\tau)=\varepsilon$.

\noindent
Let $s''$ be the state of reached by $\sigma''$: it includes a branch created by the firings
among $\{t_{Br}\}_{t\in T_{ab}}$ followed by a tree whose set vertices is $f(V_s)$ and every vertex $f(v)$ has
been created by the firing of some transition in $\{t_{v}\}_{t\in T_{ab}}$. Observe that due to our observations
on $(\N',s[r,m_0+todo])$ all other firings of $\sigma''$ are firings of transitions in $T$. By substituting
in $\sigma''$ all $t_{Br}$ by $t$ and all $t_{v}$ by $t$, one gets a firing sequence $\sigma$ of $(\N,s[r,m_0])$
with $\lambda(\sigma)=\lambda'(\sigma')$ that covers $s$.
\end{proof}

The transformation presented in the above proof can be performed in polynomial
time and this will be used in the next section.
The next proposition establishes that, as for Petri nets, coverability
does not ensure the power of ``exact counting''. The proof is interesting by itself
since it combines an argument based on WSTS (case 1) and an argument \emph{\`a la} Ogden (case 2).

\begin{proposition}
	\label{prop:ReachPNnotCovRPN}
	
	$\mathcal L_3$ is the reachability language of the Petri net of Figure~\ref{fig:PN_for_L1_and_L3} but it
	is not the coverability language of any RPN.
\end{proposition}

\begin{proof}
	\noindent
	Due to Proposition~\ref{prop:cov-inc-reach}, it is enough to prove that there does not exist
	$(\N,s[r,m_0])$ such that $\mathcal L_3=\mathcal L_R(\N,s[r,m_0],\{\emptyset\})$.
	Assume  by contradiction that there exists such  $(\N,s[r,m_0]\})$.
		For all $n$, let $\sigma_n$ be a firing sequence reaching $\emptyset$
	such that $\lambda(\sigma_n)=a^nb^nc^n$
	and $\sigma'_n$ be the prefix of $\sigma_n$ whose last transition corresponds the last occurrence of $a$.
	Denote $s_n$ the state reached by $\sigma'_n$ and the decomposition by $\sigma_n=\sigma'_n\sigma''_n$.
	Among the possible $\sigma_n$, we select one such that $s_n$ has a minimal number of threads.
	Let $Post$ be the finite set of $\nat^P$ defined by: $Post=\{W^+(t)\}_{t\in T_{ab}}$.
	
	
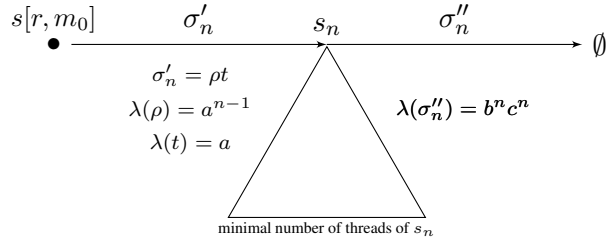
\begin{center}
		\begin{tikzpicture}[triangle/.style = {regular polygon, regular polygon sides=3 },xscale=0.45,yscale=0.45]					
			
			\path (0,3.4) node[label={[xshift=0.0cm, yshift=-0.2cm]\small{$s[r,m_0]$}}] (s0) {$\bullet$};
			\path (8,0) node[draw,minimum size=3cm,triangle,inner sep=0pt,
			label={[xshift=0.0cm, yshift=0cm]$s_n$}] (sn) {};
			\path (16,3.4) node[] (fn) {$\emptyset$};
			
			\draw[arrows=-latex'] (s0) -- (8,3.4) node[pos=0.5,above] {$\sigma'_n$} ;
			\draw[arrows=-latex'] (8,3.4) -- (fn) node[pos=0.5,above] {$\sigma''_n$} ;

			
			\path (4,2.5) node[] {\scriptsize{$\sigma'_n=\rho t$}};
			\path (4,1.5) node[] {\scriptsize{$\lambda(\rho)=a^{n-1}$}};
			\path (4,0.5) node[] {\scriptsize{$\lambda(t)=a$}};
			\path (12,1.5) node[] {\scriptsize{$\lambda(\sigma''_n)=b^{n}c^n$}};
			\path (12,1.5) node[] {\scriptsize{$\lambda(\sigma''_n)=b^{n}c^n$}};
			\path (8,-2) node[] {\tiny{minimal number of threads of $s_n$}};
			
		\end{tikzpicture}
	\end{center}

	\noindent
	{\bf Case 1.} There exists a bound $B$ of the depths of the trees corresponding to $\{s_n\}_{n\in \nat}$.
	Let $S_B$ be the set of abstract states of depth at most $B$ and different from $\emptyset$. Observe that $S_0$
	can be identified to $\nat^P$
	and  $S_B$ can be identified to $\nat^P \times {\sf Multiset}(Post \times S_{B-1})$.
	Furthermore  the (component) order on $\nat^P$ and the equality on $Post$
	are well quasi-orders. Since well quasi-order is preserved by the multiset operation and the cartesian product,
	$S_B$ is well quasi-ordered by a quasi-order denoted $<$. By construction, $s\leq s'$ implies $s\preceq_r s'$. Thus there
	exist $n<n'$ such that $s_{n}\preceq_r s_{n'}$ which entails that $\sigma'_{n'}\sigma''_{n}$ is a firing
	sequence with trace $a^{n'}b^{n}c^{n}$ reaching $\emptyset$ yielding a contradiction.

	\medskip\noindent
	{\bf Case 2.} The depths of the trees corresponding to $\{s_n\}_{n\in \nat}$ are unbounded.
	There exists $n$ such that the depth of $s_n$ is greater than $(2|Post|+1)$.
	Thus in $s_n$ for $1\leq j\leq 3$, there are edges $u_j \xrightarrow{m}_{s_n}v_j$
	and denoting $i_j$ the depth of $v_j$, one has $0<i_1<i_2<i_3$.

	\medskip\noindent
	For $k\in \{1,2,3\}$, consider of the sequence $\rho_k$ performed in the subtree rooted in $v_k$
	by the firings of $\sigma_n$. Among these three firing sequences two of them either (1)  both finish by a cut transition in $v_k$
	or (2) both do not finish by a cut transition in $v_k$. Let us call $i,j$ with $i<j$ the indices of these sequences and $w_i$  and $w_j$
	their traces. We have illustrated the situation below.
	
		
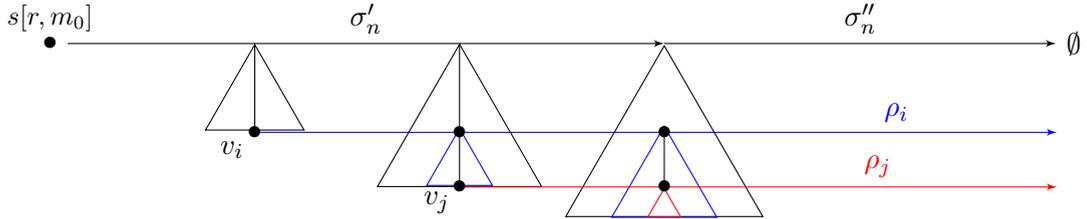
\begin{center}
		\begin{tikzpicture}[triangle/.style = {regular polygon, regular polygon sides=3 },xscale=0.45,yscale=0.45]
			\path (0,3.4) node[label={[xshift=0.0cm, yshift=-0.2cm]\small{$s[r,m_0]$}}] (s0) {$\bullet$};
			
			\path (6,1.7) node[draw,minimum size=1.5cm,triangle,inner sep=0pt,
			label={[xshift=-0.2cm, yshift=-0.7cm]{}}] (s1) {};
			\path (12,0.6) node[draw,minimum size=2.5cm,triangle,inner sep=0pt,
			label={[xshift=0.0cm, yshift=-0.7cm]{}}] (s2) {};
			\path (18,0) node[draw,minimum size=3cm,triangle,inner sep=0pt,
			label={[xshift=0.0cm, yshift=-1.5cm]{}}] (sn) {};
			
			\path (30,3.4) node[] (fn) {$\emptyset$};
			
			\path (30,0.8) node[] (fn1) {$\:\:$};
			\path (30,-0.8) node[] (fn2) {$\:\:$};
			
			\path (12,-0.2) node[draw,color=blue,minimum size=1cm,triangle,inner sep=0pt,
			label={[xshift=0.0cm, yshift=-0.7cm]{}}] (st2) {};
			
			\path (18,-0.8) node[draw,color=blue,minimum size=1.6cm,triangle,inner sep=0pt,
			label={[xshift=0.0cm, yshift=-0.7cm]{}}] (st3) {};
			
			\path (18,-1.4) node[draw,color=red,minimum size=0.5cm,triangle,inner sep=0pt,
			label={[xshift=0.0cm, yshift=-0.7cm]{}}] (su3) {};
			
			\draw[arrows=-latex'] (s0) -- (18,3.4) node[pos=0.5,above] {$\sigma'_n$} ;
			\draw[arrows=-latex'] (18,3.4) -- (fn) node[pos=0.5,above] {$\sigma''_n$} ;
			
			\draw[color=blue,arrows=-latex'] (6,0.8)  -- (fn1) node[pos=0.8,above] {$\rho_i$} ;
			\draw[color=red,arrows=-latex'] (12,-0.8)  -- (fn2) node[pos=0.7,above] {$\rho_{j}$} ;
			
			
			\draw (6,3.4) -- (6,0.8);
			\path (6,0.8) node[label={[xshift=-0.3cm, yshift=-0.7cm]$v_i$}] {$\bullet$};
			
			\draw (12,3.4) -- (12,0.8);
			\path (12,0.8) node[] {$\bullet$};
			\path (12,-0.8) node[label={[xshift=-0.3cm, yshift=-0.7cm]$v_{j}$}] {$\bullet$};
			\draw (12,0.8) -- (12,-0.8);
			
			\path (18,0.8) node[] {$\bullet$};
			\path (18,-0.8) node[] {$\bullet$};
			\draw (18,0.8) -- (18,-0.8);			
		\end{tikzpicture}
	\end{center}
	
	\vspace*{-1.8mm}\noindent
	One can build two firing sequences that still reach $\emptyset$ and thus whose labels
	belong to the language. The first one consists of mimicking the ``behavior'' of the subtree rooted
	in $v_{j}$ starting from $v_i$, which is possible due to the choice of $i$ and $j$,  as illustrated below.
	
	
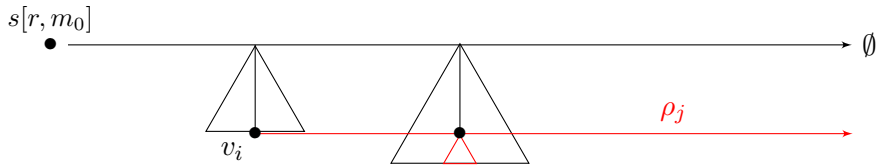
\begin{center}
		\begin{tikzpicture}[triangle/.style = {regular polygon, regular polygon sides=3 },xscale=0.45,yscale=0.45]			
			
			\path (0,3.4) node[label={[xshift=0.0cm, yshift=-0.2cm]\small{$s[r,m_0]$}}] (s0) {$\bullet$};
			
			\path (6,1.7) node[draw,minimum size=1.5cm,triangle,inner sep=0pt,
			label={[xshift=-0.2cm, yshift=-0.7cm]{}}] (s1) {};
			\path (12,1.1) node[draw,minimum size=2.1cm,triangle,inner sep=0pt,
			label={[xshift=0.0cm, yshift=-0.7cm]{}}] (s2) {};
			\path (24,3.4) node[] (fn) {$\emptyset$};
			\path (24,0.8) node[] (fn1) {$\:\:$};
		
			\path (12,0.2) node[draw,color=red,minimum size=0.5cm,triangle,inner sep=0pt,
			label={[xshift=0.0cm, yshift=-0.7cm]{}}] (su3) {};
			
		    \draw[arrows=-latex'] (s0) -- (fn) node[pos=0.5,above] {} ;
			
			\draw[color=red,arrows=-latex'] (6,0.8)  -- (fn1) node[pos=0.7,above] {$\rho_{j}$} ;
			
			
			\draw (6,3.4) -- (6,0.8);
			\path (6,0.8) node[label={[xshift=-0.3cm, yshift=-0.7cm]$v_i$}] {$\bullet$};
			
			\draw (12,3.4) -- (12,0.8);
			\path (12,0.8) node[] {$\bullet$};
		\end{tikzpicture}
	\end{center}
	
	\vspace*{-1.8mm}\noindent
	The second one consists of mimicking the ``behavior'' of the subtree rooted
	in $v_{i}$ starting from $v_{j}$ as illustrated below.
	
	
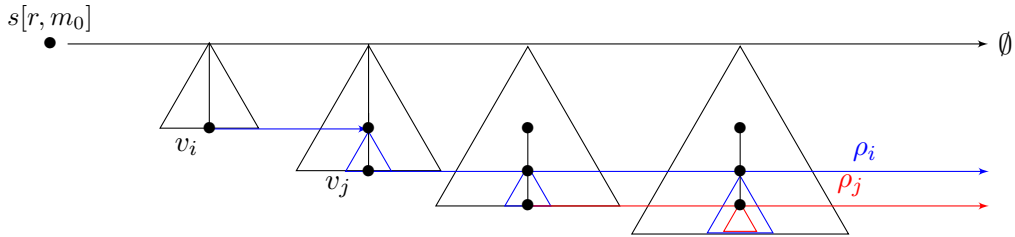
\begin{center}
		\begin{tikzpicture}[triangle/.style = {regular polygon, regular polygon sides=3 },xscale=0.35,yscale=0.35]
						
			\path (0,3.4) node[label={[xshift=0.0cm, yshift=-0.2cm]\small{$s[r,m_0]$}}] (s0) {$\bullet$};
			
			\path (6,1.3) node[draw,minimum size=1.5cm,triangle,inner sep=0pt,
			label={[xshift=-0.2cm, yshift=-0.7cm]{}}] (s1) {};
			\path (12,0.2) node[draw,minimum size=2.2cm,triangle,inner sep=0pt,
			label={[xshift=0.0cm, yshift=-0.7cm]{}}] (s2) {};
			\path (18,-0.7) node[draw,minimum size=2.8cm,triangle,inner sep=0pt,
			label={[xshift=0.0cm, yshift=-1.5cm]{}}] (sn) {};
			
			\path (26,-1.4) node[draw,minimum size=3.3cm,triangle,inner sep=0pt,
			label={[xshift=0.0cm, yshift=-1.5cm]{}}] (ssn) {};
			
			\path (36,3.4) node[] (fn) {$\emptyset$};
			\path (36,-1.4) node[] (fn1) {$\:\:$};
			\path (36,-2.7) node[] (fn2) {$\:\:$};
			
			
			\path (12,-0.9) node[draw,color=blue,minimum size=0.7cm,triangle,inner sep=0pt,
			label={[xshift=0.0cm, yshift=-0.7cm]{}}] (st2) {};
			
			\path (18,-2.2) node[draw,color=blue,minimum size=0.7cm,triangle,inner sep=0pt,
			label={[xshift=0.0cm, yshift=-0.7cm]{}}] (st2bis) {};
			
			\path (26,-3) node[draw,color=blue,minimum size=1cm,triangle,inner sep=0pt,
			label={[xshift=0.0cm, yshift=-0.7cm]{}}] (st3) {};
			
			\path (26,-3.3) node[draw,color=red,minimum size=0.5cm,triangle,inner sep=0pt,
			label={[xshift=0.0cm, yshift=-0.7cm]{}}] (su3) {};
			
		   \draw[arrows=-latex'] (s0) -- (fn) node[pos=0.5,above] {} ;
			
			\draw[color=blue,arrows=-latex'] (6,0.2)  -- (12,0.2) node[pos=0.8,above] {} ;
			
			\draw[color=blue,arrows=-latex'] (12,-1.4)  -- (fn1) node[pos=0.8,above] {$\rho_i$} ;
			\draw[color=red,arrows=-latex'] (18,-2.7)  -- (fn2) node[pos=0.7,above] {$\rho_{j}$} ;
			
			
			\draw (6,3.4) -- (6,0.2);
			\path (6,0.2) node[label={[xshift=-0.3cm, yshift=-0.7cm]$v_i$}] {$\bullet$};
			
			\draw (12,3.4) -- (12,0.2);
			\path (12,0.2) node[] {$\bullet$};
			\path (12,-1.4) node[label={[xshift=-0.4cm, yshift=-0.7cm]$v_{j}$}] {$\bullet$};
			\draw (12,0.8) -- (12,-1.4);
			
			\path (18,0.2) node[] {$\bullet$};
			\path (18,-1.4) node[] {$\bullet$};
			\path (18,-2.7) node[] {$\bullet$};
			\draw (18,0.2) -- (18,-2.7);
			
			\path (26,0.2) node[] {$\bullet$};
			\path (26,-1.4) node[] {$\bullet$};
			\path (26,-2.7) node[] {$\bullet$};
			\draw (26,0.2) -- (26,-2.7);				
		\end{tikzpicture}
	\end{center}
	
	\vspace*{-1.8mm}\noindent
	{\bf Case $w_i=w_{j}$.} Then the firing sequence reaching $\emptyset$ obtained by mimicking in $v_i$
	the behaviour of $v_{j}$ has trace $a^nb^nc^n$
	and leads to another state $s_n$ with less threads yielding a contradiction, since $s_n$
	was supposed to have a minimal number of threads.
	
	\medskip\noindent
	{\bf Case $w_i\neq w_{j}$.}  Let $w\neq \varepsilon$ be the trace of  the sequence performed in the subtree rooted in $v_i$
	without the  trace of  the sequence performed in the subtree rooted in $v_{j}$.
	Let us consider the firing sequence $\sigma$ reaching $\emptyset$ obtained by mimicking in $v_{j}$
	the behaviour of $v_{i}$. The trace of $\sigma$ is an interleaving of $a^nb^nc^n$ and $w$ and it belongs to $\mathcal L_3$
	which implies that $w=a^qb^qc^q$ for some $q>0$. Furthermore
	$\sigma$ can be chosen in such a way that the firing subsequences in the subtrees rooted at $v_i$ and $v_{j}$ are performed in one
	shot which implies that its trace is $\ldots a^qa^qw_{j}b^qc^q b^qc^q\ldots$ yielding a contradiction.	
\end{proof}
The following corollary shows that extending the family of coverability languages of PNs by substituting either (1) coverability by reachability or (2) PNs by RPNs is somewhat ``orthogonal''.
\begin{corollary}
	\label{cor:ReachPNuncompCovRPN}
	The families of reachability languages of Petri nets and
	the family of coverability languages of RPNs are incomparable.
\end{corollary}
\begin{proof}
	One direction is a consequence of Proposition~\ref{prop:ReachPNnotCovRPN}
	while the other direction is a consequence of Proposition~\ref{prop:CFinRPN}
	observing that the language of palindromes is not the reachability
	language of any Petri net.
\end{proof}
The next corollary exhibits a particular feature of RPNs languages
(e.g.  Petri nets  or context-free languages are closed under intersection with a regular language).
\begin{corollary}\label{col:intersection_rpn_and_regular}
	The family of coverability languages of RPNs is not closed under
	intersection with a regular language and under complementation.
\end{corollary}
\begin{proof}
	Due to Proposition~\ref{prop:ReachPNnotCovRPN},
	the family of coverability languages of RPNs
	is strictly included in the family of recursively enumerable languages.
	Since the former family is closed under homomorphism, Theorem~\ref{theo:re}
	implies that it is not closed under intersection with a regular language and a fortiori
	with another coverability language.
	Since intersection can be obtained by union and complementation and since the family of RPN coverability languages is closed
	under union, they are not closed under complementation.	
\end{proof}

\begin{figure}[!b]
\vspace*{-13mm}

\tikzset{every picture/.style={line width=0.75pt}} 

\begin{center}
	\begin{tikzpicture}[x=0.75pt,y=0.75pt,yscale=-1,xscale=1]
	
	\draw [line width=1.5,fill = red,fill opacity=0.7]    (190.64,191) .. controls (220.68,92.76) and (353.01,101.83) .. (389.41,191) ;
	\draw [line width=1.5,fill= blue,fill opacity=0.6]    (326.03,191) .. controls (327.65,186.83) and (329.43,182.85) .. (331.37,179.07) .. controls (375.14,93.67) and (497.95,105.62) .. (525.8,191) ;
	\draw [line width=1.5 , fill=green,fill opacity=0.1]    (50.41,190.87) .. controls (114.84,-3.81) and (338.94,14.17) .. (400.57,190.87) ;
	\draw [line width=1.5,fill=black,fill opacity=0.2]    (182.64,191) .. controls (247.07,-3.68) and (471.17,14.3) .. (532.8,191) ;
	\draw [line width=1.5,fill=yellow,fill opacity=0.05]    (20.64,191) .. controls (100.07,-80) and (471.17,-50) .. (560.8,191) ;
	\draw  [fill={rgb, 255:red, 0; green, 0; blue, 0 }  ,fill opacity=1 ] (115.29,166.37) .. controls (115.29,164.99) and (116.66,163.87) .. (118.35,163.87) .. controls (120.04,163.87) and (121.41,164.99) .. (121.41,166.37) .. controls (121.41,167.75) and (120.04,168.87) .. (118.35,168.87) .. controls (116.66,168.87) and (115.29,167.75) .. (115.29,166.37) -- cycle ;
	\draw  [fill={rgb, 255:red, 0; green, 0; blue, 0 }  ,fill opacity=1 ] (248.75,163.37) .. controls (248.75,161.99) and (250.12,160.87) .. (251.81,160.87) .. controls (253.5,160.87) and (254.87,161.99) .. (254.87,163.37) .. controls (254.87,164.75) and (253.5,165.87) .. (251.81,165.87) .. controls (250.12,165.87) and (248.75,164.75) .. (248.75,163.37) -- cycle ;
	\draw  [fill={rgb, 255:red, 0; green, 0; blue, 0 }  ,fill opacity=1 ] (425.06,161.37) .. controls (425.06,159.99) and (426.43,158.87) .. (428.12,158.87) .. controls (429.81,158.87) and (431.18,159.99) .. (431.18,161.37) .. controls (431.18,162.75) and (429.81,163.87) .. (428.12,163.87) .. controls (426.43,163.87) and (425.06,162.75) .. (425.06,161.37) -- cycle ;
	\draw  [fill={rgb, 255:red, 0; green, 0; blue, 0 }  ,fill opacity=1 ] (347.92,78.37) .. controls (347.92,76.99) and (349.29,75.87) .. (350.98,75.87) .. controls (352.67,75.87) and (354.04,76.99) .. (354.04,78.37) .. controls (354.04,79.75) and (352.67,80.87) .. (350.98,80.87) .. controls (349.29,80.87) and (347.92,79.75) .. (347.92,78.37) -- cycle ;
	
	\draw (144.8,38) node [anchor=north west][inner sep=0.75pt]  [font=\small]  {$Reach$-$PN$};
	\draw (340.41,36) node [anchor=north west][inner sep=0.75pt]  [font=\small]  {$Cov$-$RPN$};
	\draw (243.2,103) node [anchor=north west][inner sep=0.75pt]  [font=\small]  {$Cov$-$PN$};
	\draw (428.27,104) node [anchor=north west][inner sep=0.75pt]  [font=\small]  {$CF$};
	\draw (245.2,-17) node [anchor=north west][inner sep=0.75pt]  [font=\small]  {$Reach$-$RPN$};
	\draw (124.61,156) node [anchor=north west][inner sep=0.75pt]    {$\mathcal L_{3}$};
	\draw (259.07,152) node [anchor=north west][inner sep=0.75pt]    {$\mathcal L_{1}$};
	\draw (434.38,150) node [anchor=north west][inner sep=0.75pt]    {$\mathcal L_{2}$};
	\draw (361.17,67) node [anchor=north west][inner sep=0.75pt]    {$\mathcal L_{1} \cup \mathcal L_{2}$};
	
	\draw (270.17,20) node [anchor=north west][inner sep=0.75pt]    {$\mathcal L_{2} \cup\mathcal L_{3}$};
	\filldraw[black] (260,30) circle (2pt) node[] {};

	\draw[ultra thick] (19.5,191) -- (562,191);
	
	\end{tikzpicture}
\end{center}\vspace*{-1mm}
	\caption{$\mathcal L_1 =\{a^mb^nc^p\mid m\geq n\geq p \};\mathcal  L_2= \{w\in\{d,e\}^* \mid w=\widetilde{w}\};\mathcal  L_3 = \{a^nb^nc^n\mid n\in\nat\}  $}	\label{fig:lang_compare}
\end{figure}
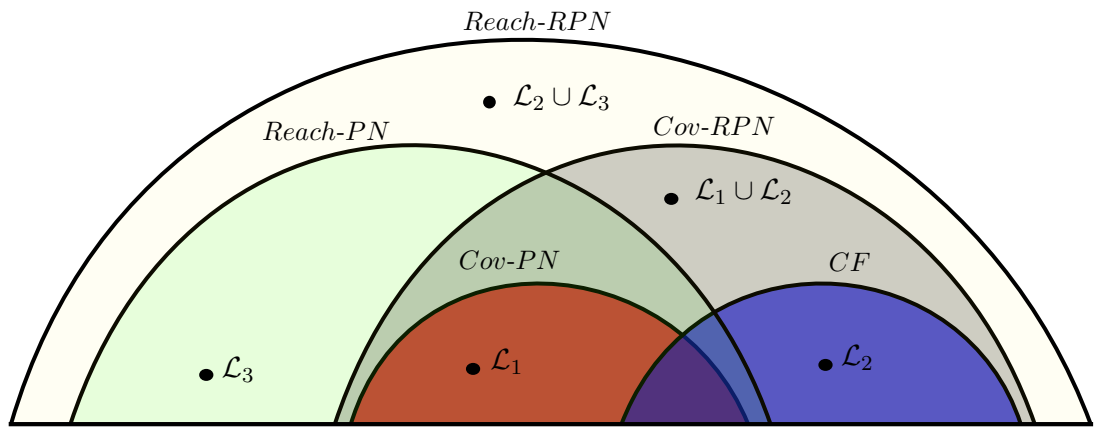

Combining Propositions~\ref{prop:reach-inc-cov},~\ref{prop:cov-inc-reach} and~\ref{prop:ReachPNnotCovRPN}, one gets the following theorem.
\begin{theorem}
	The family of coverability languages of RPNs is strictly included
	in the family of reachability languages of RPNs.
\end{theorem}
Figure~\ref{fig:lang_compare} illustrates the hierarchy of the languages presented in this work.

\section{Coverability  is \EXPSPACE -complete}
\label{sec:coverability}

The section is devoted to establishing that the coverability problem
is $\EXPSPACE$-complete. The
$\EXPSPACE$-hardness follows immediately
from the $\EXPSPACE$-hardness of the coverability problem for Petri nets~\cite{Lipton76}.

Observe that the coverability problem is equivalent to the emptiness problem of the coverability language of an RPN.
In Section \ref{sec:expressiveness} we have shown that the families of coverability languages and cut languages for RPN are equal and that the transformation from one to another is performed in polynomial time (proposition \ref{prop:reach-inc-cov} and \ref{prop:cov-inc-reach}). Therefore we will establish the complexity result for the cut problem getting as a corollary the same result for the coverability problem.

\begin{theorem}\label{thm:cut_problem}
	The cut problem is \EXPSPACE-complete.
\end{theorem}

\begin{proof}
	\noindent Let $(\N,s_0)$ be a marked RPN and $\eta$ the accumulated size of the RPN and the initial state.
	By Proposition~\ref{col:rooted} we can assume that $V_{s_0}$ is a singleton $\{r\}$.
	
	\noindent Assume there exists a firing sequence $s_{0}\xrightarrow{\sigma}_{\N}\emptyset$.  Using Proposition~\ref{prop:omniciant} one gets an omniscient sequence $ s_0\xrightarrow{\widehat{\sigma}}_{\widehat{\N}}\emptyset$ such that
	$\widehat{\sigma}=(r,\sigma_{1})(r,t)$ for some $t\in T_{\tau}$.
	
	\noindent The (omniscient) sequence $(r, \sigma_1)$  contains only elementary transitions.
	Thus $m_0 \xrightarrow{\sigma_1}_{\widehat{\N}_{el}} m$ with $m\geq W^-(t)$.
	By Theorem~\ref{thm:Rackoff covering path}, there exists  $ \sigma_1' $ with
	$ |\sigma_1'|\leq 2^{2^{c\eta\log \eta}}$ covering $W^-(t)$.	
	Using Corollary \ref{col:bound_on_translation_of_N-hat_to_N} there $s\xrightarrow{\sigma'}_{\N}\emptyset$ with $|\sigma'|\leq ^{2^{e\eta\log \eta}} $ for some constant $ e $.
	
	\noindent Therefore if there is a cut sequence
	then there is one with length at most $2^{2^{e\eta\log \eta}}$. Hence
	one guesses a sequence with at most this length and
	simultaneously checks  whether it is a cut sequence in exponential space.
	This shows that the cut problem belongs to $\NEXPSPACE$ which is equivalent to $\EXPSPACE$
	by Savitch's theorem.
	
	\noindent The \EXPSPACE\ hardness of the coverability problem in Petri nets entails \EXPSPACE\ hardness of the coverability problem in RPNs
	which in turn entails the \EXPSPACE\ hardness of the cut problem in RPNs.	
\end{proof}
The next theorem is an immediate corollary of the previous one.
\begin{theorem}\label{thm:Coverabilty for RPN in EXPSPACE}
	The coverability problem for RPNs is $\EXPSPACE$-complete.
\end{theorem}

\section{Termination is \EXPSPACE-complete}
\label{sec:termination}
In this section we tackle the termination problem for RPN.
Let $(\N,s_{0})$ be a marked RPN.
We denote the size of the input of the termination problem by $\eta$.
In~\cite{Rac78} Rackoff showed that the termination
problem for Petri net is solvable in exponential space:
\begin{theorem}[Rackoff\cite{Lipton76,Rac78}]
	\label{thm:Termination Bound For PN }The termination problem for Petri
	nets is $\EXPSPACE$-complete.
\end{theorem}
We aim to show that the termination
problem for RPN is $\EXPSPACE$-complete. $\EXPSPACE$-hardness follows
immediately from  $\EXPSPACE$-hardness of the termination problem for
Petri nets~\cite{Lipton76}.
By Proposition~\ref{col:rooted} we can assume that  $V_{s_0} = \{r\}$. Hence for the rest of the section, we will assume that $s_0 = s[r,m_0]$ for some marking $m_0$.

\medskip\noindent A main ingredient of the proof is the construction of an \emph{\at\ graph}
related to the firing of abstract transitions.

\begin{definition}[\at\ graph]
	Let $(\N,s_0)$ be a marked RPN. Let
	$ G_{\N,s_{0}}=(V_{a},E_{a},M_{a}) $ be a labeled directed graph defined inductively as follows:\smallskip
	\begin{enumerate}
		\itemsep=0.98pt
		\item $r\in V_a$ and $M_{a}(r)=m_0$;
		%
		\item For any $v\in V_a$ and $t\in T_{ab}$, if there exists  $s[v,M_{a}(v)] \xrightarrow{\sigma(v,t)}$ then
		
		$v_{t} \in V_a$, $(v,v_t)\in  E_a$ and $M_{a}(v) = \Omega(t)$.
	\end{enumerate}
	
\end{definition}
Observe that an edge $(v,v_{t})$ means that  from state $s[v,M_a(v)]$, the
thread $v$ can fire $t$ in the future
and by induction that $v_t\in V_a$ if and only if $t$ is fireable in the marked RPN.
Observe that the size of $ G_{\N,s_{0}}$ is linear w.r.t. the size of $(\N,s_0)$.
\begin{lemma}\label{lem:abstract graph in expspace}
	Let $(\N,s_{0})$ be a marked RPN.
	Then one can build its abstract graph in exponential space.
\end{lemma}
\begin{proof}
	First note that $|V_a|\leq|T_{ab}|+1$.
	Then for any vertex $v$ already in $V_a$ and any $t\in T_{ab}$
	checking whether  $s[v,M_{a}(v)] \xrightarrow{\sigma(v,t)}$ is fireable
	is equivalent to solving the covering problem
	$M_{a}(v)\xrightarrow{\sigma} m\succeq W^-(t)$ in $\widehat{\N}_{el}$ (recall Definition~\ref{def:N_el})
	which can be done in exponential space due to Rackoff's coverability theorem for Petri nets.
\end{proof}
While we will not prove it, using a reduction from the Petri net coverability problem,
one can show that we cannot use less than an exponential space to build the abstract graph.

\medskip
Let us illustrate the \at\ graph in
Figure~\ref{fig:at_graph} corresponding to the RPN of Figure~\ref{fig:rpn_example}.
Here the initial state is $s[r,p_{ini}]$. For clarity, we have renamed the abstract transitions
as follows:
$
t:={t_{beg}}$, $ta:={t_{a_2}}$, $tb:={t_{b_2}}.
$
For instance, the existence of the edge from $v_t$ to $v_{ta}$
is justified by the firing sequence $(v_t,t_{a_1})(v_t,ta)$.

\begin{figure}[!h]
	\begin{center}
			\begin{tikzpicture}[scale=0.7,
            > = stealth, 
            shorten > = 1pt, 
            auto,
            semithick 
        ]
        \tikzstyle{state}=[
        	circle,
            draw = black,
            thick,
            fill = white,
            minimum size = 4mm
        ]

        
        \node[state, label=\small $r$] (v1) at (0,0)  {};
        \node[state ,label=90:\small $v_{t}$] (v2) [right = 15mm  of v1] {};
        \node[state,label=175:\small $v_{ta}$] (v3) [above right = 2mm and 15mm of v2] {};
        \node[state,label=185:\small $v_{tb}$] (v4) [below right = 2mm and 15mm of v2] {};
        
        \path[->] (v1) edge node {} (v2);
        \path[->] (v2) edge node {} (v3);
        \path[->] (v2) edge node  {}  (v4);
        \path[->, bend left] (v3) edge node  {}  (v4);
        \path[->, bend left] (v4) edge node  {}  (v3);
        \path[->] (v3) edge[loop right ] node  {}  (v3);
        \path[->] (v4) edge[loop right ] node  {}  (v4);

\end{tikzpicture}
	\end{center}\vspace*{-6mm}
	\caption{An abstract graph for the RPN in
		Figure~\ref{fig:rpn_example} }
	\label{fig:at_graph}
\end{figure}
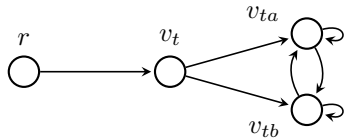

Let $\sigma$ be an infinite firing sequence. We say that $\sigma$
is \emph{\lseq}\ if it visits a state $s$ whose depth is strictly greater
than $|T_{ab}|$. Otherwise, we say that $\sigma$ is \emph{\sseq}.
To solve the termination problem it suffices
to show whether the RPN has such an infinite sequence, either \sseq\
or \lseq.

The next lemma establishes that
\emph{lassos} of the \at\ graph
are witnesses of \lseq\ infinite sequences in an RPN:
\begin{restatable}{lemma}{unboundedpath}
	\label{lem:Finidng unbounded path using at}
	Let $(\N,s_{0})$ be a marked RPN.
	Then there is a \lseq\ infinite sequence starting from $s_{0}$
	if and only if there is a cycle in $G_{\N,s_{0}}$.
\end{restatable}
\begin{proof}
	$\bullet$ Assume that $ \sigma $ is a \lseq\ sequence. Hence, it reaches
	a state  $\tilde{s}$ whose tree has a path $ \gamma $
	starting from the root, with $|\gamma|> |T_{ab}|$. Let us denote it by
	$\gamma=(v_{i})_{i=1}^{m}$.
	For all $ i\leq m $ denote by $ t_i $ the abstract transition that creates
	$ v_i $. Using $ \gamma$, one builds a path
	$ \gamma_{a}=v_{1}v_{2}\ldots v_{m} $ in $G_{\N,s_{0}}$ as follows. First
	$ v_{1}=r $ and $ m_r= M_{a}(r) $. Since along
	$ \sigma $ the
	thread $ r $ fires  $ t_1 $ to create $v_2$, there is an
	edge between $ r $ to $ v_{t_2} $ in $ G_{\N,s_{0}} $. For any
	$ 1< i\leq m $ the thread $ v_{i} $ is created with the marking
	$ \Omega(t_i)=M_{a}(v_{t_i}) $. Since $ v_{i+1} $ is a child of $ v_i $,
	somewhere on the sequence $ \sigma $ the thread $ v_i $ fires
	$ t_{i+1} $. Therefore there is an edge from $ v_{t_{i}} $
	to $ v_{t_{i+1}} $ in $ G_{\N,s_{0}}$.
	The length of the path $ \gamma_a $ strictly greater then $ |T_{ab}|$,  and since $V_{a}\leq|T_{ab}|+1$ there is a cycle in $\gamma_a$.
	
	\noindent $\bullet$
	Conversely assume that there is a cycle in $G_{\N,s_{0}}$.
	Then there is an
	infinite path $\gamma_a=\{ v_{i}\} _{i=0}^{\infty}$
	in $G_{\N,s}$ starting from $r$, where for any $i\geq1$ denote by $t_{i}$ the abstract
	transition associated the vertex $v_i$. We now translate this infinite path
	to an \lseq\ sequence on $\N$ with initial state $s_{0}$. Note that
	$v_{0}=r$ and that $m_r=M_{a}(r)$. By definition of
	$E_{a}$ there is a sequence $s \xrightarrow{\sigma_{1}}s_0'$ where the
	abstract transition $ t_1 $ is fireable from $v_{0}$ in $ s_0' $. We get
	$ s\xrightarrow{\sigma_{1}}s_0'\xrightarrow{(v_0,t_1)}s_{2}$. Denote by
	$ v_1 $ the thread created by $t_1$. The threads marking has
	$M_{s_1}(v_1)=M_{a}(v_1)$, therefore one continues translating the path
	$ \gamma_a $ in the same way as the first edge. Since for any
	$ (v_i,v_{i+1}) $ in $ \gamma_{a} $ we create a new thread from
	$ v_i $ one gets an \lseq\ sequence.
\end{proof}
We now show that for any \sseq\  $\sigma$ there
is a thread $v$ which fires infinitely many times in $\sigma$.

\begin{restatable}{lemma}{infinitelymanytimes}
	\label{lem: thread repets infinitly many times}
	Let $(\N,s_0)$ be a marked RPN and $\sigma$ be a \sseq\
	sequence. Then there is a thread $v$ that fires infinitely many times
	in $\sigma$.
\end{restatable}
\begin{proof}
	If the root $r$ fires infinitely often then we are done.
	Otherwise, $r$ has finitely many children, and the firing subsequence of $\sigma$ of the subtree of (at least) one child, say $v$, must be infinite.
	If $v$ fires infinitely often then we are done.
	Otherwise, we proceed inductively up to $|T_{ab}|$ where some thread must fire infinitely often.
\end{proof}

We now show that given some state $s[r,m_0]$ one can check
in exponential space the existence
of a \sseq\ sequence in which $r$ fires infinitely many times.

\begin{restatable}{lemma}{sseqFromu}
	\label{lem:sseq from a given u}
	Let  $(\N,s_0)$ be a marked RPN. Then  one can check in exponential space,
	whether there exists an infinite sequence starting
	with $r$ firing infinitely many times.
\end{restatable}
\begin{proof}
	We first show that there is a sequence  where $r$ fires infinitely many times
	if and only if there is a infinite firing sequence in the marked Petri net $(\widehat{\N}_{el},m_0)$.
	
	\noindent
	$\bullet$ Assume there exists such $\sigma$ in $(\N,s[r,m_0])$. Then the sequence $\sigma$ is also fireable in $(\widehat{\N},s[r,m_0])$.
	In $\widehat{\N}$, one eliminates in $\sigma$ the cut transitions by increasing occurrence order as follows. Let $(v,t)$ be a cut transition
	and $(v',t')$ be the firing that creates $v$. Then one deletes all the firings performed by the descendants
	of $v$ and replaces $(v',t')$ by $(v',t'^r)$. Let $\sigma'$ be the sequence obtained after this transformation.
	In $\sigma'$, the root still fires infinitely often since no firing performed by the root has been deleted (but sometimes
	substituted by an elementary firing). Moreover, $\sigma'$ has no more cut transitions.
	Consider the still infinite firing sequence $(r,\sigma'')$ where in  $\sigma'$ all firings in other vertices than $r$ have been deleted.
	Observe now  that by definition, $\sigma''$ is also an infinite sequence of $\widehat{\N}_{el}$.
	
	\noindent
	$\bullet$  Conversely, assume there exists an infinite firing sequence $\sigma$ of $(\widehat{\N}_{el},m_0)$. Then
	$(r,\sigma)$  is  an infinite firing sequence of $(\widehat{\N},s[r,m_0])$ (with only root firings)
	entailing the existence of  an infinite firing sequence of $(\N,s[r,m_0])$.
	
	\noindent
	By Theorem~\ref{thm:Termination Bound For PN }, one can check in exponential space
	whether there exists an infinite sequence of $(\widehat{\N}_{el},m_0)$.
\end{proof}

Summing up the results for \sseq\ and \lseq\ sequences we get:

\begin{theorem}
	\label{prop:Bounded path in EXPSPACE}
	The termination problem of RPN is \EXPSPACE-complete.
\end{theorem}

\begin{proof}
	The algorithm proceeds as follows.
	It builds in \EXPSPACE\ (by Lemma~\ref{lem:abstract graph in expspace}) the \at\ graph and
	checks  whether
	there is a \lseq\ infinite sequence
	using the characterization of Lemma~\ref{lem:Finidng unbounded path using at}.
	In the negative case, it looks for a \sseq\ infinite
	sequence. To this aim, it checks in exponential space for any reachable vertex $v$
	from $r$ in $G_{\N,s_{0}}$, whether
	there exists an infinite sequence starting from $s[v,M_{a}(v)]$
	with the root firing infinitely many times.
	The complexity follows from Lemma~\ref{lem:sseq from a given u}
	while the correctness
	follows from Lemma~\ref{lem: thread repets infinitly many times}.
\end{proof}

\section{Finiteness and boundedness are \EXPSPACE-complete}
\label{sec:finiteness}
In this section we will show that the finiteness and boundedness problems for RPNs are $\EXPSPACE$-complete w.r.t. $\eta=size(\N,s_{0})$, i.e. the accumulated size of the RPN and the initial state.
For Petri nets the finiteness problem, which is equivalent to the boundedness problem, has been shown to be $\EXPSPACE$-complete:
\begin{theorem}[\cite{Lipton76,Rac78}]
	\label{thm:Finiteness Bound For PN }The finiteness problem for Petri
	nets is $\EXPSPACE$-complete.
\end{theorem}
$\EXPSPACE$-hardness follows immediately from  $\EXPSPACE$-hardness of the finiteness problem for Petri nets~\cite{Lipton76}.

Moreover by applying Proposition~\ref{col:rooted} like in previous sections we will assume that $s_0=s[r,m_0]$.
Given two vertices $u,v$ in a graph $\mathcal G$, the \emph{distance} between them $dist_{\mathcal G}(u,v)$ is
the length of a shortest path going from one to the other.

\begin{lemma} \label{lem:Form a path in ab to a path in RPN}
	Let $(\N,s_0)$ be a marked RPN and $G_{\N,s_{0}}=(V_a,E_a,M_a)$ be its abstract graph. Then for all $v\in V_a$, there exists $s\in Reach(\N, s_{0})$ and $u\in V_s $ such that $ M_s(u)=M_a(v)$.
\end{lemma}

\begin{proof}
	We show the lemma by induction on $dist_{G_{\N,s_{0}}}(r,u)$. If $dist_{G_{\N,s_{0}}}(r,v) =0$ then $v=r$ and $M_a(r)=m_0$.	
	Assume that we have shown the lemma for any $v$ such that $dist_{G_{\N,s_{0}}}(r,v)<n$, and pick $v\in V_a$ such that $dist_{G_{\N,s_{0}}}(r,v)=n$. Since $dist_{G_{\N,s_{0}}}(v,r)>0$, $v=v_t$ for some $t\in T_{ab}$. Moreover there is some $(u,v_t)\in E_a$ such that $dist_{G_{\N,s_{0}}}(r,u)=n-1$ and by the induction hypothesis there is a sequence $s_{0}\xrightarrow{\sigma_u}s_u$ and some $w\in V_{s_u}$ such that $M_{s_u}(w)=M_a(u)$. From the definition of $G_{\N,s_{0}}$ there is a fireable sequence $ s[w,M_a(w)]\xrightarrow{\sigma_t(w,t)}$. Combining these sequences, we get $s_{0}\xrightarrow{\sigma_u}s_u\xrightarrow{\sigma_t(w,t)}s_{v_t}$, where the newly created thread $w'$ fulfills $M_{s_v}(w')=\Omega(t)=M_a(v_t)$.
\end{proof}
The following lemma shows that we can simulate the behaviour of every thread by a Petri net.
\begin{lemma} \label{lem:every marking is covered by ab graph}
	Let $(\N,s_{0})$ be a marked RPN and $G_{\N,s_{0}}=(V_a,E_a,M_a)$ be its abstract graph.
	Then:
	$$
	\bigcup_{s\in Reach(\N,s_{0})}\{M_s(v) \}_{v\in V_s}= \bigcup_{u\in V_a} Reach(\widehat{\N}_{el}, M_a(u)).
	$$
\end{lemma}

\begin{proof}	
	\noindent $\bullet$ Let $m\in\bigcup_{s\in Reach(\N,s_{0})}\{M_s(u) \}_{u\in V_s} $. There exists $s_0\xrightarrow{\sigma}s$ with some $v\in V_s $ such that  $M_s(v)=m$. By Proposition~\ref{prop:omniciant} there is an omniscient sequence in $s_0\xrightarrow{\widehat{\sigma}}_{\widehat{\N}}s$. We split $\widehat\sigma$ into $s_0\xrightarrow{\widehat{\sigma}_1}_{\widehat{\N}}s_v\xrightarrow{\widehat{\sigma}_2}_{\widehat{\N}}$ where $s_v$ is the the state where the thread $v$ first appears. Note that there is $u\in V_a$ for which $M_{s_v}(v) = M_a(u)$. Let $(v,\widehat{\sigma}_2')$ consisting of all firings
	of $v$ in $\widehat{\sigma}_2$. $(v,\widehat{\sigma}_2')$ is fireable from $s_v$ since $\widehat{\sigma}_2$ is omniscient implying that there will be not cut transition fired by a child of $v$. By construction of $\widehat{\N}_{el}$, the sequence $\widehat{\sigma}_2'$ is a firing sequence
	of $(\widehat{\N}_{el},M_a(u))$ thus $ m\in Reach(\widehat{\N}_{el}, M_a(u))$.
	
	\noindent $\bullet$ Let $u\in V_a $ and $m\in Reach(\widehat{\N}_{el}, M_a(u))$, i.e. $M_a(u)\xrightarrow{\sigma}_{\widehat{\N}_{el}}m$ for some $n\in \nat$. First by Lemma~\ref{lem:Form a path in ab to a path in RPN} there exists $s_0\xrightarrow{\sigma_u}_\N s_u$ where for some $v\in V_{s_u}$ we have $M_{s_u}(v) = M_a(u)$. By construction of $\widehat{\N}$ we also have $s_0\xrightarrow{\sigma_u}_{\widehat{\N}} s_u$. By  construction of $\widehat{\N}_{el}$ we get that $s_u\xrightarrow{(v,\sigma)}_{\widehat{\N}}s$ where $M_s(v) = m$. By Proposition~\ref{prop:equivreach}, $s\in Reach(\N,s_{0})$, which concludes the proof.	
\end{proof}
Using the previous Lemma and Rackoff's Theorem we establish the complexity of the boundedness problem:
\begin{proposition}
	The boundedness problem of RPN is \EXPSPACE-complete.
\end{proposition}
\begin{proof}
	Hardness of the problem comes from hardness of Petri nets.
	Let $(\N,s_0)$ be a marked RPN. First by Corollary~\ref{col:rooted} we can assume that $s_0 = s[r,m_0]$.
	By Lemma \ref{lem:every marking is covered by ab graph} checking whether $\N,s_0$ is bounded
	is equivalent to whether for $v \in V_a$, $(\widehat{\N}_{el}, M_a(u))$ is bounded which, due to Rackoff,
	can be performed in exponential space.	
\end{proof}

Let $(\N,s_{0})$ be a marked RPN. If $s_{0}=\emptyset$ then the number of reachable states is finite (one), hence from now on we  assume that $s_{0}\neq\emptyset$. Next, if there exists $t\in T_{ab}$ with $W^-(t)=0$ then there are infinitely many reachable states since one can fire $t$ repeatedly which provides us with a sequence of states with an unbounded number of threads. Therefore from now on we  assume that for all $t\in T_{ab}$, $W^-(t)>0$.

We now establish a connection between the boundedness of $\widehat{\N}_{el}$ and the maximal number of children of the root in $\N$:
\begin{lemma}\label{lem:bounded num of children}
	Let $\N$ be an RPN such that $(\widehat{\N}_{el},m_0)$ is bounded. Then:
	\[
	\sup_{s'\in Reach(\N,s[r,m_0])}|\{v\in V_{s'}\mid r_{s'}\rightarrow_{s'} v\}|< \infty
	\]
\end{lemma}

\begin{proof}
	Assume that there exists a family of sequences $ \{\sigma_n\}_{n\in\nat} $ such that $s[r,m_0]\xrightarrow{\sigma_n}_{\N}s_n$ and the number of children of $r$ in $s_n$ is greater than $n$. By Proposition~\ref{prop:omniciant} for all $\sigma_n$ there exists an omniscient sequence $\widehat{\sigma}_n$ in $\widehat{\N}$ from $s[r,m]$ reaching $s_n$. We remove from $\widehat{\sigma}_n$ all the transitions not fired from the root getting $(r,\widehat{\sigma}_n')$ which is also fireable from $s[r,m]$ and which leads to a state where the root has a number of children greater than $n$. Since an abstract transition consumes tokens from the root (for all $t\in T_{ab}$, $W^-(t)>0$) one can remove them from  $(r,\widehat{\sigma}_n')$ and get $(r,\widehat{\sigma}_n'')$ for which $s\xrightarrow{(r,\widehat{\sigma}_n'')}_{\widehat{\N}}s_n''$ and $\sum_{p\in P}M_{s_n''}(r)(p)>n$.
	Since $ \widehat{\sigma}_n''$ is fireable from $m$ in $\widehat{\N}_{el}$ this contradicts the hypothesis of the lemma.	
\end{proof}
Combining the above results, we get a characterization of the  finiteness problem:

\begin{proposition}\label{prop:const finitness is equvalent to}
	Let $(\N,s_{0})$ be a marked RPN. Then $Reach(\N, s_{0})$ is finite if and only if both of the following assertions hold:\\
	1. There is no loop in $G_{\N,s_{0}}=(V_a,E_a,M_a)$;\\
	2. For all $v\in V_a$, $(\widehat{\N}_{el},M_a(v))$ is bounded.
\end{proposition}
\begin{proof}
	$\bullet$ Assume that assertions 1 and 2 hold. Due to Assertion~1 and Lemma~\ref{lem:Finidng unbounded path using at} any reachable state has its depth bounded by some constant. Due to Assertion~2 and Lemmas~\ref{lem:every marking is covered by ab graph} and~\ref{lem:bounded num of children} each thread in any reachable state has a bounded number of children, and a bounded number of different reachable markings. Therefore $Reach(\N,s_0)$ is finite.\smallskip	
	
	\noindent $\bullet$ Assume that Assertion~1 does not hold. By Lemma~\ref{lem:Finidng unbounded path using at} there is a deep infinite sequence. Hence there is an infinite sequence of states with growing depth. Therefore $Reach(\N,s_0)$ is not finite.
	
	\noindent $\bullet$ Assume that Assertion~2 does not hold for some vertex $v$. By Lemma~\ref{lem:Form a path in ab to a path in RPN} there exists a state $s\in Reach(\N, s_{0})$ and a vertex $u\in V_s $ such that $ M_s(u)=M_a(v)$. By the definition of $\widehat{\N}_{el}$, for any $m\in Reach(\widehat{\N}_{el},M_s(v))$, there exists a firing sequence $(r,\sigma')$ in $\widehat{\N}$ such that $s\xrightarrow{(r,\sigma')}_{\widehat{\N}}s'$ with $M_{s'}(v)=m$. Therefore $Reach(\widehat{\N},s)\subseteq Reach(\widehat{\N},s_0)$. Due to Proposition~\ref{prop:equivreach}, $Reach(\widehat{\N},s_0)= Reach(\N,s_0)$.	
\end{proof}

\begin{theorem}
	\label{prop:constraind finitness in EXPSPACE}
	The finiteness problem of RPN is \EXPSPACE-complete.
\end{theorem}

\begin{proof}
	The algorithm proceeds by checking Assertions~1 and~2 of Proposition\ref{prop:const finitness is equvalent to}.
	It builds in exponential space (by Lemma~\ref{lem:abstract graph in expspace}) the \at\ graph and
	checks whether there is no loop in $G_{\N,s_{0}}$.
	In the negative case, it checks in exponential space for any vertex $v\in V_a$, whether
	the marked Petri net $(\widehat{\N}_{el},M_{a}(v))$
	is bounded.	
\end{proof}

\section{Conclusion}
\label{sec:conclusion}


%
We have proven that RPN is a strict generalisation of both Petri nets
and context-free grammars without increasing the complexity of coverability, termination, boundedness and finiteness problems.
It remains several open problems about languages of RPN and decidability/complexity
of checking properties. Here is a partial list of open problems:
\begin{itemize}
\item  How to decide whether a word belongs
to a coverability or reachability language of a RP?
%
\item Since the quasi-order possesses an infinite antichain, but there exist short witnesses for coverability,
does there exist an effective finite representation
of the downward closure of the reachability set?
\item Does there exist a relevant fragment of LTL decidable for RPN?
\end{itemize}

  \subsection{Acknowledgment}
  We thank the reviewers very much for their deep, detailed and insightful reviews, which helped us a lot in order to simplify and clarify this paper.


\begin{thebibliography}{10}
\providecommand{\url}[1]{\texttt{#1}}
\providecommand{\urlprefix}{URL }
\expandafter\ifx\csname urlstyle\endcsname\relax
  \providecommand{\doi}[1]{doi:\discretionary{}{}{}#1}\else
  \providecommand{\doi}{doi:\discretionary{}{}{}\begingroup
  \urlstyle{rm}\Url}\fi
\providecommand{\eprint}[2][]{\url{#2}}

\bibitem{Mayr84}
Mayr EW.
\newblock An Algorithm for the General Petri Net Reachability Problem.
\newblock \emph{{SIAM} J. Comput.}, 1984.
\newblock \textbf{13}(3):441--460.

\bibitem{abs-1809-07115}
Czerwinski W, Lasota S, Lazic R, Leroux J, Mazowiecki F.
\newblock The reachability problem for {P}etri nets is not elementary.
\newblock In: Proceedings of {STOC} 19. 2019 pp. 24--33.
doi:10.1145/3313276.3316369.

\bibitem{DBLP:conf/lics/LerouxS19}
Leroux J, Schmitz S.
\newblock Reachability in Vector Addition Systems is Primitive-Recursive in
  Fixed Dimension.
\newblock In: Proceedings of {LICS} 19. 2019 pp. 1--13.
doi:10.1109/LICS.2019.8785796.

\bibitem{Rac78}
Rackoff C.
\newblock The covering and boundedness problems for vector addition systems.
\newblock \emph{Theoretical Computer Science}, 1978.
\newblock \textbf{6}(2):223 -- 231.

\bibitem{Reinhardt08}
Reinhardt K.
\newblock Reachability in {P}etri Nets with Inhibitor Arcs.
\newblock \emph{Electr. Notes Theor. Comput. Sci.}, 2008.
\newblock \textbf{223}:239--264. doi:10.1016/j.entcs.2008.12.042.

\bibitem{BFLZ-lmcs12}
Bonnet R, Finkel A, Leroux J, Zeitoun M.
\newblock Model Checking Vector Addition Systems with one zero-test.
\newblock \emph{LMCS}, 2012.
\newblock \textbf{8}(2:11).  doi:10.2168/LMCS-8 (2:11).

\bibitem{Bonnet11}
Bonnet R.
\newblock The Reachability Problem for Vector Addition System with One Zero-Test.
\newblock In: {MFCS} 2011, Warsaw, Poland, volume 6907 of \emph{LNCS}. 2011 pp.
  145--157.  doi:10.1007/978-3-642-22993-0\_16.

\bibitem{PhS-mfcs10}
Schnoebelen {\relax Ph}.
\newblock Revisiting {A}ckermann-Hardness for Lossy Counter Machines and Reset {P}etri Nets.
\newblock In: {MFCS} 2010, Brno, Czech Republic, volume 6281 of \emph{LNCS}.
  2010 pp. 616--628.   doi:10.1007/978-3-642-15155-2\_54.

\bibitem{dufourd98}
Dufourd C, Finkel A, Schnoebelen {\relax Ph}.
\newblock Reset Nets between Decidability and Undecidability.
\newblock In: {ICALP}'98, volume 1443 of \emph{LNCS}. Springer, Aalborg,
  Denmark, 1998 pp. 103--115.  doi:10.1007/ BFb0055044.

\bibitem{lazic:hal-01265302}
Lazi{\'c} R, Schmitz S.
\newblock {The Complexity of Coverability in $\nu$-Petri Nets}.
\newblock In: {LICS 2016}. {ACM Press}, New York, United States, 2016 pp.
  467--476.   doi:10.1145/2933575.2933593.

\bibitem{Lazic13}
Lazic R.
\newblock The reachability problem for vector addition systems with a stack is
  not elementary.
\newblock \emph{CoRR}, 2013. \newblock \textbf{abs/1310.1767}. \newblock \eprint{1310.1767}.

\bibitem{LazicS14}
Lazic R, Schmitz S.
\newblock Non-elementary complexities for branching VASS, MELL, and extensions.
\newblock In: {CSL-LICS} 2014, Vienna, Austria. {ACM}, 2014 pp. 61:1--61:10.
	doi:10.1145/2733375.

\bibitem{jcss12-DJLL}
Demri S, Jurdzi{\'n}ski M, Lachish O, Lazi{\'c} R.
\newblock The covering and boundedness problems for branching vector addition  systems.
\newblock \emph{Journal of Computer and System Sciences}, 2012.
\newblock \textbf{79}(1):23--38.   doi:10.1016/ j.jcss.2012.04.002.

\bibitem{AtigG11}
Atig MF, Ganty P.
\newblock Approximating {P}etri Net Reachability Along Context-free Traces.
\newblock In: {FSTTCS} 2011, Mumbai, India, volume~13 of \emph{LIPIcs}. 2011
  pp. 152--163.  doi:10.4230/LIPIcs.FSTTCS.2011.152.

\bibitem{MavlankulovOTSZ18}
Mavlankulov G, Othman M, Turaev S, Selamat MH, Zhumabayeva L, Zhukabayeva T.
\newblock Concurrently controlled grammars.
\newblock \emph{Kybernetika}, 2018.
\newblock \textbf{54}(4):748--764.  doi:10.14736/kyb-2018-4-0748.

\bibitem{DassowT09}
Dassow J, Turaev S.
\newblock Petri Net Controlled Grammars: the Case of Special Petri Nets.
\newblock \emph{J. {UCS}}, 2009.
\newblock \textbf{15}(14):2808--2835.

\bibitem{Zetzsche15}
Zetzsche G.
\newblock The Emptiness Problem for Valence Automata or: Another Decidable
  Extension of {P}etri Nets.
\newblock In: {RP} 2015, Warsaw, Poland, volume 9328 of \emph{LNCS}. 2015 pp.
  166--178.

\bibitem{EFH-icmas96}
El~Fallah~Seghrouchni A, Haddad S.
\newblock A Recursive Model for Distributed Planning.
\newblock In: {ICMAS} 1996, Kyoto, Japan. 1996 pp. 307--314.
ISBN:978-1-57735-013-2.

\bibitem{HP-icatpn99}
Haddad S, Poitrenaud D.
\newblock Theoretical Aspects of Recursive {P}etri Nets.
\newblock In: {ICATPN} 1999, Williamsburg, Virginia, USA, volume 1639 of
  \emph{LNCS}. 1999 pp. 228--247.  doi:10.1007/3-540-48745-X\_14.

\bibitem{haddad:hal-01573071}
Haddad S, Poitrenaud D.
\newblock {Modelling and Analyzing Systems with Recursive {P}etri Nets}.
\newblock In: {WODES} 2000, Ghent, Belgium, volume 569 of \emph{The Springer
  International Series in Engineering and Computer Science}. 2000 pp. 449--458.
  doi:10.1007/978-1-4615-4493-7\_48.

\bibitem{HaddadP01}
Haddad S, Poitrenaud D.
\newblock Checking Linear Temporal Formulas on Sequential Recursive {P}etri Nets.
\newblock In: TIME 2001, Civdale del Friuli, Italy. {IEEE} Computer Society,
  2001 pp. 198--205.  doi:10.1109/ TIME.2001.930718.

\bibitem{HaddadP07}
Haddad S, Poitrenaud D.
\newblock Recursive {P}etri nets.
\newblock \emph{Acta Inf.}, 2007.
\newblock \textbf{44}(7-8):463--508.  doi:10.1007/s00236-007-0055-y.

\bibitem{FHK-atpn19}
Finkel A, Haddad S, Khmelnitsky I.
\newblock Coverability and Termination in Recursive Petri Nets.
\newblock In: {PETRI~NETS}'19, volume 11522 of \emph{LNCSs}. Springer, Aachen,
  Germany, 2019 pp. 429--448.  HAL Id: hal-02081019, URL \url{https://hal.inria.fr/hal-02081019}.

\bibitem{Stadel78}
Stadel M.
\newblock A remark on the time complexity of the subtree problem.
\newblock \emph{Computing}, 1978.
\newblock \textbf{19}(4):297--302.

\bibitem{Alain01}
Finkel A, Schnoebelen P.
\newblock Well-structured transition systems everywhere!
\newblock \emph{Theor. Comput. Sci.}, 2001.
\newblock \textbf{256}(1-2):63--92.  doi:10.1016/S0304-3975(00)00102-X.

\bibitem{DBLP:conf/stoc/CzerwinskiLLLM19}
Czerwinski W, Lasota S, Lazic R, Leroux J, Mazowiecki F.
\newblock The reachability problem for Petri nets is not elementary.
\newblock In: {STOC} 2019. {ACM}, 2019 pp. 24--33.  arXiv:1809.07115 [cs.FL].

\bibitem{nla.cat-vn2956435}
Peterson JL.
\newblock Petri net theory and the modeling of systems / James L. Peterson.
\newblock Prentice-Hall Englewood Cliffs, N.J, 1981.
\newblock ISBN:0136619835.

\bibitem{DBLP:journals/acta/GeeraertsRB07}
Geeraerts G, Raskin J, Begin LV.
\newblock Well-structured languages.
\newblock \emph{Acta Informatica}, 2007.
\newblock \textbf{44}(3-4):249--288.   doi:10.1007/s00236-007-0050-3.

\bibitem{BFHR-icomp13}
Bonnet R, Finkel A, Haddad S, Rosa{-}Velardo F.
\newblock Ordinal Theory for Expressiveness of Well-Structured Transition  Systems.
\newblock \emph{Information and Computation}, 2013.
\newblock \textbf{224}:1--22.   doi:10.1016/j.ic.2012.11.003.

\bibitem{DBLP:journals/tcs/DelzannoR13}
Delzanno G, Rosa{-}Velardo F.
\newblock On the coverability and reachability languages of monotonic  extensions of Petri nets.
\newblock \emph{Theor. Comput. Sci.}, 2013.
\newblock \textbf{467}:12--29. doi:10.1016/j.tcs.2012.09.021.

\bibitem{DBLP:journals/jcss/ValkV81}
Valk R, Vidal{-}Naquet G.
\newblock Petri Nets and Regular Languages.
\newblock \emph{J. Comput. Syst. Sci.}, 1981.
\newblock \textbf{23}(3):299--325.

\bibitem{figueira:hal-02193089}
Figueira D.
\newblock {Co-finiteness of VASS coverability languages}, 2019.
\newblock Working paper or preprint,
  \urlprefix\url{https://hal.archives-ouvertes.fr/hal-02193089}.

\bibitem{DBLP:conf/rp/HofmanT14}
Hofman P, Totzke P.
\newblock Trace Inclusion for One-Counter Nets Revisited.
\newblock In: {RP} 2014, volume 8762 of \emph{LNCS}. Springer, 2014 pp. 151--162.
doi:10.1007/978-3-319-11439-2\_12.

\bibitem{HaddadP99}
Haddad S, Poitrenaud D.
\newblock Decidability and undecidability results for recursive {P}etri nets.
\newblock Technical Report 019, LIP6, Paris VI University, 1999.
Id: hal-02548232, URL \url{https://hal.archives-ouvertes.fr/hal-02548232}.

\bibitem{ogden1968helpful}
Ogden W.
\newblock A helpful result for proving inherent ambiguity.
\newblock \emph{Mathematical systems theory}, 1968.
\newblock \textbf{2}(3):191--194. doi:10.1007/BF01694004.

\bibitem{Lambert92}
Lambert J.
\newblock A Structure to Decide Reachability in {P}etri Nets.
\newblock \emph{Theor. Comput. Sci.}, 1992.
\newblock \textbf{99}(1):79--104. doi:10.1016/0304-3975(92)90173-D.

\bibitem{Lipton76}
Lipton RJ.
\newblock The Reachability Problem Requires Exponential Space.
\newblock Technical Report 062, Yale University, Department of Computer
  Science, 1976.
\end{thebibliography}
\end{document}